\documentclass[10pt]{iopart}

    \usepackage{iopams}
    \usepackage{amssymb}
    \usepackage{amsfonts}
    \usepackage{amstext}
    \usepackage{graphicx}
    \usepackage{setstack}
    \usepackage{amscd}
    \usepackage{color}
		\usepackage{ulem}

    \newcommand{\ihbar}{\imath \hbar}
    \newcommand{\Te}{\mathbb{T}e}
    \newcommand{\Pe}{\mathbb{P}_{\mathcal C} e}
    \newcommand{\llangle}{\langle \hspace{-0.2em} \langle}
    \newcommand{\rrangle}{\rangle \hspace{-0.2em} \rangle}
    \renewcommand{\S}{\mathcal H_{\mathcal S}}
    \newcommand{\E}{\mathcal H_{\mathcal E}}

    \newcommand{\Ad}{\mathrm{Ad}}

    \newcommand{\Sp}{\mathrm{Sp}}
    \newcommand{\im}{\Im\mathrm{m}}
    \newcommand{\re}{\Re\mathrm{e}}
    \newcommand{\Obj}{\mathrm{Obj}}
    \newcommand{\Morph}{\mathrm{Morph}}

    \newenvironment{proof}{\noindent \textit{Proof:}}{\hfill $\Box$ \\}

    \newtheorem{defi}{Definition}
    \newtheorem{prop}{Property}
    \newtheorem{propo}{Proposition}
    \newtheorem{theo}{Theorem}

\begin{document}

    \title[Geometric phases in open quantum systems and higher gauge theory]{A new kind of geometric phases in open quantum systems and higher gauge theory}

    \author{David Viennot and Jos\'e Lages}

    \address{Institut UTINAM (CNRS UMR 6213, Universit\'e de Franche-Comt\'e), 41 bis Avenue de
    l'Observatoire, BP1615, 25010 Besan{\c c}on cedex, France}

		\eads{
		\mailto{david.viennot@utinam.cnrs.fr},
		\mailto{jose.lages@utinam.cnrs.fr}
		}

\begin{abstract}
A new approach is proposed, extending the concept of geometric phases to adiabatic open quantum systems described by density matrices (mixed states). This new approach is based on an analogy between open quantum systems and dissipative quantum systems which uses a $C^*$-module structure. The gauge theory associated with these new geometric phases does not employ the usual principal bundle structure but a higher structure, a categorical principal bundle (so-called principal 2-bundle or non-abelian bundle gerbes) which is sometimes a non-abelian twisted bundle. The need to site the gauge theory in this higher structure is a geometrical manifestation of the decoherence induced by the environment on the quantum system.
\end{abstract}

\section{Introduction}
Since the pioneering work of Berry \cite{berry} concerning the adiabatic dynamics of a closed two-level system, several types of geometric phases for quantum systems have been studied, including non-abelian geometric phases \cite{wilczek}, nonadiabatic geometric phases \cite{aharonov}, geometric phases associated with noncyclic evolution \cite{samuel}, geometric phases associated with the Floquet theory \cite{moore}, and geometric phases associated with resonances \cite{mondragon, mostafazadeh,viennot2}. The geometric structures within which the Berry phases arise have also been studied, e.g., the line bundle associated with an abelian Berry phase \cite{simon}, the principal bundle associated with a non-abelian Berry phase \cite{wilczek}, and the composite bundle associated with a Berry phase which does not commute with the dynamical phase \cite{sardana,viennot1,viennot4}.\\
These geometric phase phenomena are related to pure states $\psi$ (vectors in a Hilbert space $\mathcal H$) governed by the Schr\"odinger equation $\ihbar \frac{d\psi}{dt} = H(t)\psi(t)$ (where $H \in \mathcal L(\mathcal H)$ is the Hamiltonian, which is selfadjoint for the closed quantum systems but not for the dissipative quantum systems). The open quantum systems \cite{breuer,bengtsson, bratteli} are described by mixed states (density matrices $\rho \in \mathcal L(\mathcal H)$, $\rho^\dagger = \rho$, $\rho \geq 0$ and $\tr_{\mathcal H} \rho = 1$) which are governed by a Lindblad equation $\ihbar \frac{d \rho}{dt} = \mathcal L(\rho)$ (where the Lindbladian $\mathcal L$ is often written in the standard form: $\mathcal L(\rho) = [h,\rho] - \frac{\imath}{2} \{\Gamma^\dagger_k \Gamma^k,\rho\} + \imath \Gamma^k \rho \Gamma^\dagger_k$ , $h,\Gamma^k \in \mathcal L(\mathcal H)$, $\{.,.\}$ being the anticommutator). Four approaches have previously been proposed to extend the concept of the geometric phases to open quantum systems. Uhlmann's approach \cite{uhlmann1,uhlmann2,uhlmann3,uhlmann4} is based on the purification of mixed states; since a purified state is Hilbert-Schmidt, we can use the definition of the geometric phases for vectors of a Hilbert space. The second approach \cite{chaturvedi} is based on the decomposition of mixed states into convex combinations of pure states. As the geometric phase associated with a pure state is the holonomy of the natural connection of a universal bundle over a complex projective space, this generalization to mixed states is obtained by considering the natural connection of the bundle induced by the convex combinations. The third approach due to Sarandy and Lidar \cite{sarandy1,sarandy2} is based on the fact that in finite dimension ($\dim \mathcal H = n$) the reduced density matrices $\rho-I_n$ are Hilbert-Schmidt. Thus the Lindblad equation can be considered as a Schr\"odinger equation in a $n^2$-dimensional Hilbert space and we can use the definition of the geometric phases for vectors of a Hilbert space. The last approach, due to Sj\"oqvist \textit{et al.} \cite{sjoqvist,tong}, is based on the introduction of a notion phase for evolving mixed states in interferometry. This method is generalized and extended to a nonunitary evolution via a purification method.\\
In this paper we propose another approach to define the geometric phases of an open quantum system based on an analogy between open and dissipative quantum systems. This approach uses the identification of a mixed state with a ``norm'' in a $C^*$-module (see for example ref. \cite{landsman}). This identification provides a new interpretation of the geometric phases of open quantum systems. The geometric structure in which this new kind of geometric phase is sited is not a principal bundle (as used for the geometric phases of a closed quantum system) but a higher gauge structure, a nonabelian bundle gerbes \cite{breen, aschieri, kalkkinen} or equivalently a principal 2-bundle \cite{baez1,baez2,baez3}, i.e. a categorical generalization of a principal bundle. The higher degree in the gauge theory seems to be a manifestation of the decoherence induced by the environment on which the quantum system is open. In previous works, we have already shown that such a higher degree arises for geometric phases of some quantum systems which are described by pure states but which present a kind of decoherence process due to resonances crossings \cite{viennot2} or to the use of the adiabatic Floquet theory \cite{viennot3}.\\
Section 2 is a brief review concerning geometric phases of closed and dissipative quantum systems. In section 3, we develop an analogy between dissipative quantum systems and open quantum systems, which allows us to introduce in section 4 a new kind of geometric phase for open quantum systems and its related geometric structure. Section 5 presents the adiabatic approximation for open quantum systems from the viewpoint of the geometric phase approach introduced in section 4. Section 6 presents a very simple example: the control of a qubit subjected to a decoherence process. For the sake of simplicity the present paper focuses on finite dimensional open quantum systems.

\section{Geometric phases of closed and dissipative quantum systems}
\subsection{Geometric phases of closed quantum systems}
Let $(\mathcal H,H(x),M)$ be a quantum dynamical system, where $\mathcal H$ is an Hilbert space representing the quantum states space (for the sake of simplicity we suppose a finite dimensional Hilbert space, $\dim \mathcal H = n$), $M \ni x \mapsto H(x) \in \mathcal L(\mathcal H)$ is a continuous family of self-adjoint operators representing the driven Hamiltonian, and $M$ is a $\mathcal C^\infty$-manifold representing the space of all possible configurations of the external classical parameters controlling the quantum system. A given dynamics is generated by a control parameter evolution $t \mapsto x(t)$ as being the solution of the Schr\"odinger equation:
\begin{equation}
\ihbar \frac{d \psi(t)}{dt} = H(x(t)) \psi(t) \qquad \psi(0) = \psi_0 \in \mathcal H
\end{equation}
Let $\{U^\alpha\}_\alpha$ be a good open cover of $M$ (a set of contractible open sets such that $\bigcup_\alpha U^\alpha = M$). Let $M \ni x \mapsto e(x) \in \mathbb R$ be a continuous eigenvalue of $H(x)$ supposed non-degenerate and $U^\alpha \ni x \mapsto \phi_e^\alpha(x) \in \mathcal H$ be the associated (locally defined) continuous normalized eigenvector (the eigenvector is only locally defined because in many cases there do not exist eigenvectors continuously extendable on the whole of $M$).
\begin{equation}
\forall x \in U^\alpha, \quad H(x) \phi_e^\alpha(x) = e(x) \phi_e^\alpha(x)
\end{equation}
If we suppose that $\psi(0) = \phi_e^\alpha(x(0))$ and that the evolution $t \mapsto x(t)$ is slow with respect to the quantum proper time of state transition, i.e.
\begin{eqnarray}
& & \forall f \in \Sp(H) \setminus \{e\}, \forall \phi_f \in \ker(H-f 1_{\mathcal H}) \nonumber \\
& & \quad  \langle \phi_f(x(t)) | \frac{d}{dt} \phi_e(x(t)) \rangle = \frac{\langle \phi_f(x(t))|\frac{dH(x(t))}{dt}|\phi_e(x(t)) \rangle}{e(x(t))-f(x(t))} \simeq 0
\end{eqnarray}
(note that this requires a gap condition $|e(x(t)) - f(x(t))| > 0$, $\forall t$, $\forall f \in \Sp(H) \setminus\{e\}$), then the wave function is
\begin{equation}
\label{adiabtranspusuelle}
\psi(t) \simeq e^{-\ihbar^{-1}\int_0^t e(x(t'))dt'} e^{-{\int_{x(0)}^{x(t)}}_{\mathcal C} A^\alpha } \phi_e^\alpha(x(t))
\end{equation}
where $\mathcal C$ is the path in $M$ drawing by $t \mapsto x(t)$ and
\begin{equation}
A^\alpha = \langle \phi_e^\alpha|d\phi_e^\alpha \rangle \in \Omega^1 (U^\alpha,\imath \mathbb R)
\end{equation}
is the generator of the geometric phase ($d$ is the exterior derivative of $M$ and $\Omega^n (M,\imath \mathbb R)$ is the set of pure imaginary valued differential $n$-forms of $M$). We have supposed that $\mathcal C$ is completely included in $U^\alpha$. Equation (\ref{adiabtranspusuelle}) is called the adiabatic transport formula, $e^{-\ihbar^{-1}\int_0^t e(t')dt'}$ is called the dynamical phase and $e^{-{\int_{x(0)}^{x(t)}}_{\mathcal C} A^\alpha }$ is the geometric phase which depends only on the shape of $\mathcal C$ in $M$ and not on the time duration. By construction the wave function $\tilde \psi(t) = e^{-{\int_{x(0)}^{x(t)}}_{\mathcal C} A^\alpha} \phi_e^\alpha(x(t))$ satisfies the parallel transport condition $\langle \tilde \psi(t)|\frac{d\tilde \psi(t)}{dt} \rangle = 0$.\\
If $\mathcal C$ crosses several charts of $\{U^\alpha\}_\alpha$ we have
\begin{eqnarray}
\label{alvaformule}
\psi(t) & \simeq & e^{-\ihbar^{-1} \int_0^t e(x(t'))dt'} e^{-{\int_{x(0)}^{x^{\alpha \beta}}}_{\mathcal C} A^\alpha} g^{\beta \alpha}(x^{\alpha \beta}) e^{-{\int_{x^{\alpha \beta}}^{x^{\beta \gamma}}}_{\mathcal C} A^\beta}... \nonumber \\
& & ... g^{\xi \zeta}(x^{\zeta \xi}) e^{-{\int_{x^{\xi \zeta}}^{x(t)}}_{\mathcal C} A^\zeta} \phi_e^\zeta(x(t))
\end{eqnarray}
where $U^\alpha,U^\beta,...,U^\xi,U^\zeta$ are the different charts crossed by $\mathcal C$, $x^{\alpha \beta}$ is an arbitrary point in $U^\alpha \cap U^\beta \cap \mathcal C$ and $g^{\alpha \beta}(x) \in U(1)$ is defined by
\begin{equation}
\forall x \in U^\alpha \cap U^\beta \quad \phi_e^\beta(x) = g^{\alpha \beta}(x) \phi_e^\alpha(x)
\end{equation}
(since over $U^\alpha \cap U^\beta$, $\phi_e^\alpha$ and $\phi_e^\beta$ are two normalized eigenvectors of the same non-degenerate eigenvalue, they differ only by a phase factor $g^{\alpha \beta}$). The formula of the geometric phase is independent of the arbitrary choices of transition points $\{x^{\alpha \beta}\}_{\alpha,\beta}$ \cite{alvarez}. The transition functions $g^{\alpha \beta}$ satisfy the cocycle relation:
\begin{equation}
\forall x \in U^\alpha \cap U^\beta \cap U^\gamma \quad g^{\alpha \beta}(x) g^{\beta \gamma}(x) g^{\gamma \alpha}(x) = 1
\end{equation}
which is required in order to have a single valued eigenvector $\phi^\alpha_e = g^{\alpha \beta}g^{\beta \gamma} g^{\gamma \alpha} \phi^\alpha_e$.

As pointed out by Simon \cite{simon} the geometric phase is associated with the horizontal lift of $\mathcal C$ in $P$ the principal $U(1)$-bundle over $M$ defined by the transition functions $g^{\alpha \beta}$ and endowed with the connection having $A^\alpha$ as gauge potential (see ref. \cite{nakahara} for an exposition of the principal bundle theory). If the path $\mathcal C$ is closed, we have at the end $T$ of the evolution
\begin{equation}
\psi(T) \simeq e^{-\ihbar^{-1} \int_0^T e(x(t'))dt'} \mathrm{hol}(\mathcal C) \phi_e^\alpha(x(0))
\end{equation}
where $\mathrm{hol}(\mathcal C)$ is the holonomy of the horizontal lift of $\mathcal C$ in $P$ which (if $\mathcal C \subset U^\alpha$) is
\begin{equation}
\mathrm{hol}(\mathcal C) = e^{\oint_{\mathcal C} A^\alpha}
\end{equation}
The geometry of the principal bundle $P$ is characterized by the curvature $F = dA^\alpha \in \Omega^2(M,\imath \mathbb R)$ which satisfies $dF = 0$ (since $dA^\alpha = dA^\beta$ for all $x\in U^\alpha \cap U^\beta$ the curvature $F$ is globally defined).\\

We now suppose that $e_a$ is a $m$-fold degenerate eigenenergy on the whole of $M$. Let $\{\phi_{e,a}^\alpha \}_{a=1,...,m}$ be an orthonormal basis of continuous associated eigenvectors. If $\psi(0) = \phi_{e,a}^\alpha(x(0))$ and $\mathcal C \subset U^\alpha$ then at the adiabatic limit we have
\begin{equation}
\psi(t) \simeq e^{-\ihbar^{-1} \int_0^t e(x(t'))dt'} \sum_{b=1}^m \left[\Pe^{- \int_{x(0)}^{x(t)} A^\alpha} \right]_{ba} \phi_{e,b}^\alpha(x(t))
\end{equation}
where
\begin{equation}
A^\alpha = \left(\begin{array}{ccc} \langle \phi_{e,1}^\alpha|d\phi_{e,1}^\alpha \rangle & ... & \langle \phi_{e,1}^\alpha|d\phi_{e,m}^\alpha \rangle \\ \vdots & \ddots & \vdots \\ \langle \phi_{e,m}^\alpha | d\phi_{e,1}^\alpha \rangle & ... & \langle \phi_{e,m}^\alpha |d\phi_{e,m}^\alpha \rangle \end{array} \right) \in \Omega^1(U^\alpha,\mathfrak u(m))
\end{equation}
with $\Omega^n(U^\alpha,\mathfrak u(m))$ the set of $\mathfrak u(m)$ valued $n$-differential forms of $M$, and $\mathfrak u(m)$ the Lie algebra of the antiselfadjoint matrices of order $m$, and where $\Pe^{- \int_{x(0)}^{x(t)} A^\alpha}$ is the path ordered exponential along $\mathcal C$ \cite{nakahara}, i.e. $\Pe^{-\int_{x(0)}^{x(t)} A}$ is solution of
\begin{equation}
\frac{d\Pe^{- \int_{x(0)}^{x(t)} A}}{dt} = -A_\mu(x(t)) \dot x^\mu(t) \Pe^{- \int_{x(0)}^{x(t)} A}
\end{equation}
The non-abelian geometric phase $\Pe^{- \int_{x(0)}^{x(t)} A^\alpha}$ is then associated with the horizontal lift of $\mathcal C$ in $P$ the principal $U(m)$-bundle on the right over $M$ ($U(m)$ is the Lie group of the unitary matrices of order $m$) defined by the transition functions $g^{\alpha \beta}(x) \in U(m)$ such that
\begin{equation}
\forall x \in U^\alpha \cap U^\beta, \quad \left[g^{\alpha \beta}(x)\right]_{ab} = \langle \phi^\alpha_{e,a}(x)|\phi^\beta_{e,b}(x) \rangle
\end{equation}
$g^{\alpha \beta}$ is the passage matrix between the two bases $\{\phi_{e,a}^\alpha\}_a$ and $\{\phi_{e,a}^\beta\}_a$. $P$ is endowed with the connection having $A^\alpha$ as gauge potential, and for a closed path $\Pe^{\oint A^\alpha}$ is the holonomy of the horizontal lift of $\mathcal C$ in $P$. The geometry of $P$ is characterized by the curvature $F^\alpha = dA^\alpha + A^\alpha \wedge A^\alpha \in \Omega^2(U^\alpha,\mathfrak u(m))$ which satisfies $F^\beta(x) = g^{\alpha \beta}(x)^{-1} F^\alpha(x) g^{\alpha \beta}(x)$ ($\forall x \in U^\alpha \cap U^\beta$) and the Bianchi identity $dF^\alpha + [A^\alpha,F^\alpha]=0$.

\subsection{Geometric phases of dissipative quantum systems}
Let $(\mathcal H, H(x),M)$ be a quantum dynamical system where $H(x)$ is not selfadjoint with $\frac{1}{2\imath} (H(x)-H(x)^\dagger) \leq 0$. The system is then a dissipative quantum system, the eigenvalues of $H(x)$ are complex with a negative imaginary part (they are called resonances). Let $t \mapsto \psi(t)$ be the wave function solution of the Schr\"odinger equation. The non-selfadjointness of $H$ implies
\begin{equation}
\label{densdiss}
\|\psi(0)\|^2 = 1 \qquad \forall t>0, \quad \|\psi(t)\|^2<1
\end{equation}
The decreasing of $\|\psi(t)\|^2$ is characteristic of the dissipation of the quantum system.\\
Let $M \ni x \mapsto e(x) \in \mathbb C$ be a continuous supposed non-degenerate eigenvalue of $H(x)$ and $U^\alpha \ni x \mapsto \phi_e^\alpha(x) \in \mathcal H$ be the associated (locally defined) continuous (not necessary normalized) eigenvector. In a similar way to that for the closed case, we can define at the adiabatic limit a ``dissipative geometric phase\footnote{Since the dissipative geometric phase is not unitary, another terminology, as geometric factor for example, could be more appropriate.}'' $e^{-{\int_{x(0)}^{x(t)}}_{\mathcal C}A^\alpha} \in \mathbb C^*$ by defining the geometric phase generator as
\begin{equation}
A^\alpha = \frac{\langle \phi_e^\alpha|d\phi_e^\alpha \rangle}{\| \phi_e^\alpha \|^2} \in \Omega^1(U^\alpha,\mathbb C)
\end{equation}
Dissipative geometric phases are associated with the horizontal lifts in a principal $\mathbb C^*$-bundle $P$.\\
The gauge group $\mathbb C^*$ must be viewed as the product $\mathbb C^* = U(1) \times \mathbb R^{*+}$ of two physical gauge groups (and consequently $P$ must be viewed as the product of two bundles). $U(1)$ is associated with the ordinary geometric phase acquired by the wave function during the evolution, which is completely equivalent to the geometric phase of a closed quantum system, whereas $\mathbb R^{*+}$ is associated with the geometric contribution to the dissipation. If we are interested only in the dissipation process, we can describe just the structure involving the gauge group $\mathbb R^{*+}$.\\
The main object is then $\|\psi\|^2 \in \mathbb R^{*+}$ which plays the role of the state of dissipation of the quantum system. We call it the density of the system. The Schr\"odinger equation induces the following dynamical equation for the dissipation:
\begin{equation}
\ihbar \frac{d\|\psi\|^2}{dt} = \mathcal L_{x(t)}(\|\psi\|^2)
\end{equation}
with
\begin{equation}
\mathcal L_x(\|\psi\|^2) = \langle \psi |(H(x)-H(x)^\dagger)|\psi \rangle
\end{equation}
The eigenvalue equation induces the following equation for the eigendensity
\begin{equation}
\mathcal L_x(\|\phi_e^\alpha(x)\|^2) = 2 \imath \im(e(x)) \|\phi_e^\alpha(x)\|^2
\end{equation}
The generator of the density geometric phase (the geometric contribution to the dissipation) is defined as
\begin{equation}
\re A^\alpha = \frac{1}{2} \frac{d\|\phi_e^\alpha\|^2}{\|\phi_e^\alpha\|^2} \in \Omega^1(U^\alpha,\mathbb R)
\end{equation}
$\re A^\alpha$ is then the gauge potential of the connection of a principal $\mathbb R^{*+}$-bundle.\\
The geometry of the density geometric phases is trivial since the connection is pure gauge, i.e. the gauge potential is $d$-exact:
\begin{equation}
\re A^\alpha = d \ln (\|\phi_e^\alpha\|)
\end{equation}
Consequently, each holonomy is reduced to one:
\begin{equation}
\mathrm{hol}(\mathcal C) = e^{\oint_{\mathcal C} \re A^\alpha} = e^{\oint_{\mathcal C} d \ln(\|\phi_e^\alpha\|)} = 1
\end{equation}
Although the geometry of the density geometric phases of dissipative quantum systems is trivial, we use it as a model to construct a theory for open quantum systems which is not trivial. 

\section{$C^*$-module of an open quantum system, eigenvectors and gauge invariances}
\subsection{The $C^*$-module modelling an open quantum system}
An open quantum system consists of a quantum system $\mathcal S$ which can exchange matter and/or energy with an environment $\mathcal E$. The total closed system $\mathcal S+\mathcal E$ is called the universe. Being a closed system, the universe can be described as a quantum dynamical system $(\S \otimes \E, H(x), M)$ where $\S$ is the Hilbert space of the quantum states of the system (which is supposed finite dimensional, $\dim \S = n$) and $\E$ is the Hilbert space of the quantum states of the environment. To avoid any confusion, the scalar product in $\S$ is denoted by $\langle .|.\rangle_{\S}$, the scalar product in $\E$ is denoted by $\langle .|. \rangle_{\E}$ and the scalar product in $\S \otimes \E$ is denoted by$\llangle .|. \rrangle$. In the same manner, the adjoint with respect to $\langle . |.\rangle_{\S}$ is denoted by $\dagger$ whereas the adjoint with respect to $\llangle .|. \rrangle$ is denoted by $\ddagger$. The Hamiltonian of the universe $H(x) \in \mathcal L(\S \otimes \E)$ (which is selfadjoint for $\llangle .|.\rrangle$) can be written as
\begin{equation}
H(x) = H_{\mathcal S}(x) \otimes 1_{\E} + 1_{\S} \otimes H_{\mathcal E}(x) + H_I(x)
\end{equation}
$H_{\mathcal S} \in \mathcal L(\S)$ is the Hamiltonian of the isolated system, $H_{\mathcal E} \in \mathcal L(\E)$ is the Hamiltonian of the environment and $H_I \in \mathcal L(\S \otimes \E)$ is the system-environment interaction operator (all these operators are selfadjoint for the scalar products of their respective Hilbert spaces). The mixed state $\rho$ of the system associated with the normalized pure state $\psi$ of the universe is defined as being the partial trace of the orthogonal projector on $\psi$:
\begin{equation}
\label{densmat}
\rho = \tr_{\E}(|\psi \rrangle \llangle \psi|)
\end{equation}
$\rho$ is called a density matrix and we denote by $\mathcal D(\S) = \{\rho \in \mathcal L(\S)| \rho^\dagger = \rho, \rho \geq 0, \tr_{\S} \rho = 1 \}$ the set of density matrices. $\rho$ can be viewed as a statistical convex combination of pure states, $\rho = \sum_{i=1}^n p_i |\chi_i \rangle \langle \chi_i|$ where $\{p_i\}_i = \Sp(\rho)$ ($0 \leq p_i \leq 1$) and $\{\chi_i \in \S\}_i$ are the associated eigenvectors. The Schr\"odinger equation for the environment induces the Schr\"odinger-von Neuman equation
\begin{equation}
\ihbar \frac{d}{dt} |\psi(t) \rrangle \llangle \psi(t) | = [H(x(t)),|\psi(t)\rrangle \llangle \psi(t)|]
\end{equation}
which induces by the partial trace the Lindblad equation
\begin{equation}
\ihbar \frac{d\rho(t)}{dt} = \mathcal L_{x(t)}(\rho(t))
\end{equation}
where
\begin{eqnarray}
\mathcal L_x(\rho) & = & [H_{\mathcal S}(x),\rho] + \tr_{\E}\left([H_{\mathcal E}(x)+H_I(x),|\psi \rrangle \llangle \psi|]\right) \nonumber \\
& = & [H_{\mathcal S}(x),\rho] + \mathcal D_x(\rho)
\end{eqnarray}
The Liouvillian $[H_{\mathcal S}(x),.]$ represents the free evolution of the system, whereas $\mathcal D_x$ represents the decoherence induced by the environment on the system.\\
A dissipative quantum system can be viewed as a kind of open quantum system in which the dissipation is the only one effect of the coupling with the environment. The density $\|\psi\|^2 \in \mathbb R^{+*}$ of equation (\ref{densdiss}) and the density matrix $\rho \in \mathcal D(\S)$ of equation (\ref{densmat}) play then the same role. They characterize the influence of the environment on the system. The equivalent of $\S$ for a dissipative quantum system is $\mathbb C$, the equivalent of $\E$ is then $\mathcal H$ and the equivalent of the state space of the universe $\S \otimes \E$ is then $\mathbb C \otimes \mathcal H = \mathcal H$.\\
We can then consider the open quantum systems in the same manner as the dissipative quantum systems, and use the following paradigm: we consider $\S \otimes \E$ no longer as a vector space over the ring $\mathbb C$ but as a left $C^*$-module over the $C^*$-algebra $\mathcal L(\S)$ (a module which has the same axioms as a vector space but where an algebra takes the place of $\mathbb C$, see ref. \cite{landsman}). The $C^*$-module $\S \otimes \E$ is endowed with the following inner product:
\begin{equation}
\begin{array}{rcl}
(\S \otimes \E) \times (\S \otimes \E) & \to & \mathcal L(\S) \\
(\psi,\phi) & \mapsto & \langle \psi|\phi \rangle_* = \tr_{\E} \left(|\phi \rrangle \llangle \psi| \right)
\end{array}
\end{equation}
This inner product has the following useful properties:
\begin{itemize}
\item it is linear on the right and antilinear on the left: $\forall A,B \in \mathcal L(\S), \forall \phi,\psi,\chi \in \S \otimes \E$
\begin{equation}
\langle \psi|A\phi+B\chi \rangle_* = A \langle \psi|\phi \rangle_* + B \langle \psi|\chi\rangle_*
\end{equation}
\begin{equation}
\langle A\psi+B\chi|\phi \rangle_* = \langle \psi|\phi \rangle_* A^\dagger + \langle \chi|\phi\rangle_* B^\dagger
\end{equation}
\item it is ``hermitian'': $\forall \phi,\psi \in \S \otimes \E$
\begin{equation}
\langle \phi|\psi \rangle_* = \langle \psi | \phi \rangle_*^\dagger
\end{equation}
\item it is ``positive definite'': $\forall \psi \in \S\otimes\E$
\begin{equation}
\langle \psi|\psi \rangle_* \in \mathcal D(\S) \times \mathbb R^{*+}
\end{equation}
\begin{equation}
\langle \psi | \psi \rangle_* = 0 \iff \psi = 0
\end{equation}
\end{itemize}
A mixed state of the system (a density matrix) is then the square $*$-norm of a normalized state of the universe:
\begin{equation}
\rho = \|\psi\|^2_*
\end{equation}
The analogy between open and dissipative quantum systems is then complete.\\
Remark : The Hamiltonian of a dissipative system is not selfadjoint; for an open system, even if $H$ is selfadjoint for the scalar product $\llangle.|.\rrangle$, it is not selfadjoint for the inner product $\langle .|. \rangle_*$.

\subsection{Eigenoperator equation in the $C^*$-module}
The replacement of $\mathbb C$ by $\mathcal L(\S)$ requires the introduction of a new definition of the eigenvectors. 
\begin{defi}[Eigenoperator and $*$-eigenvector]
$M \ni x \mapsto E(x) \in \mathcal L(\S)$ is said to be a continuous eigenoperator of $H(x)$ and $U^\alpha \ni x \mapsto \phi^\alpha_E(x) \in \S \otimes \E$ is said to be an associated (locally defined) continuous $*$-eigenvector if
\begin{equation}\label{eoe}
H(x) \phi_E^\alpha(x) = E(x) \phi_E^\alpha(x)
\end{equation}
with
\begin{equation}
[E(x) \otimes 1_{\E}, H(x)] = 0
\end{equation}
$x \mapsto E(x)$ being a continuous map on $M$ and $x \mapsto \phi_E^\alpha(x)$ being a continuous map on $U^\alpha$.
\end{defi}
This definition is the exact transposition of an eigenvalue equation except that $E$ is an operator belonging to the $C^*$-algebra $\mathcal L(\S)$. We note that the commutation between $E$ and $H$ is required to have a behaviour sufficiently close to an eigenvalue. We suppose moreover that $\llangle \phi_E^\alpha | \phi_E^\alpha \rrangle = 1$ in order to ensure that $\rho_E^\alpha = \|\phi_E^\alpha\|^2_* = \tr_{\E}(|\phi_E^\alpha \rrangle \llangle \phi_E^\alpha|)$ is a mixed state (the mixed eigenstate). The eigenoperator satisfies some interesting mathematical properties which are set out in \ref{annexeA}; one of them permits an easy solution of the eigenoperator equation.

\subsection{Gauge invariances of the eigenoperator equation}
For the usual eigenvalue equation $H \phi_e = e \phi_e$ (with $e\in \mathbb C$ non degenerate) the gauge invariance is very simple, the action of the group $\mathbb C^*$ (or $U(1)$ if we want to preserve the normalization) transforms an eigenvector associated with $e$ to another eigenvector associated with $e$. The situation is more complicated for the eigenoperator equation (\ref{eoe}).\\
Let $G_x \subset \mathcal{GL}(\S)$ be the maximal subgroup of $\mathcal{GL}(\S)$ ($\mathcal{GL}(\S)$ is the group of invertible operators of $\S$) which leaves invariant the vector subspace $\ker(H(x)-E(x) \otimes 1_{\E}) \subset \S \otimes \E$ :
\begin{equation}
G_x \ker(H(x)-E(x)\otimes 1_{\E}) \subset \ker(H(x)-E(x)\otimes 1_{\E})
\end{equation}
Let $K_x \subset \mathcal U(\E)$ be the isotropy subgroup of $H(x)$ within $\mathcal U(\E)$ ($\mathcal U(\E)$ is the group of unitary operators of $\E$) :
\begin{equation}
K_x = \{k \in \mathcal U(\E)| k^{-1} H(x) k = H(x) \} \equiv \mathcal U(\E)_{H(x)}
\end{equation}
$G_x \times K_x$ constitutes the gauge group leaving invariant the eigenoperator equation, $\forall gk \in G_x \times K_x$ we have
\begin{equation}
H(x)\phi_E^\alpha(x) = E(x)\phi_E^\alpha(x) \Rightarrow H(x)gk\phi_E^\alpha(x) = E(x)gk\phi_E^\alpha(x)
\end{equation}
This follows directly from the definitions of $G_x$ and $K_x$ and from the fact that $kE=Ek$ since $k \in \mathcal L(\E)$ and $E \in \mathcal L(\S)$. \\
We denote by $\mathfrak g_x$ the Lie algebra of $G_x$ (the maximal subalgebra of $\mathcal L(\S)$ leaving invariant $\ker(H(x)-E(x)\otimes 1_{\E})$) and by $\mathfrak k_x =\{k \in \mathcal L(\E)|k^\dagger = -k, [1_{\S}\otimes k,H(x)] = 0\}$ the Lie algebra of $K_x$.

\begin{prop}
Under a gauge transformation $\tilde \phi_E^\alpha(x) = gk \phi_E^\alpha(x)$ with $gk \in G_x \times K_x$, the mixed eigenstate is tranformed as follows:
\begin{equation}
\tilde \rho_E^\alpha(x) = g \rho_E^\alpha(x) g^\dagger
\end{equation}
\end{prop}

\begin{proof}
\begin{eqnarray}
\tilde \rho_E & = & \tr_{\E} (|\tilde \phi_E \rrangle \llangle \tilde \phi_E |) \nonumber \\
& = & \tr_{\E} (gk|\phi_E\rrangle\llangle\phi_E|k^{-1}g^\dagger) \nonumber \\
%& = & g \tr_{\E}(k|\phi_E\rrangle\llangle\phi_E|k^{-1})g^\dagger \\
& = & g \tr_{\E}(k^{-1}k|\phi_E \rrangle\llangle \phi_E|) g^\dagger \nonumber \\
& = & g \rho_E g^\dagger
\end{eqnarray}
\end{proof}
 
$G_x$ is the maximal subgroup of $\mathcal{GL}(\S)$ leaving equivariant the square $*$-norm. From the viewpoint of the analogy between open and dissipative quantum systems, a $G_x$-gauge transformation is equivalent to a norm change of the eigenvector. The property for the dissipative quantum systems, $\forall \lambda \in \mathbb C^*$, $\tilde \phi_e^\alpha = \lambda \phi_e^\alpha \Rightarrow \|\tilde \phi_e^\alpha\|^2 = |\lambda|^2 \|\phi_e^\alpha \|^2$, becomes for the open quantum systems, $\forall g \in G_x$, $\tilde \phi_E^\alpha = g \phi_E^\alpha \Rightarrow \|\tilde \phi_E^\alpha \|^2_* = g \|\phi_E^\alpha \|^2_* g^\dagger$. Since $K_x$ leave invariant the square $*$-norm, a $K_x$-gauge transformation is a kind of ``phase change'' specific to the open quantum system.\\
The definition of $G_x \times K_x$ the group equivalent to a norm change for a $*$-eigenvector permits us to define the notion of degeneracy for the eigenoperator.

\begin{defi}[Non-degenerate eigenoperator]
Let $M \ni x \mapsto E(x) \in \mathcal L(\S)$ be an eigenoperator of $H(x)$ and $U^\alpha \ni x \mapsto \phi^\alpha_E(x) \in \S \otimes \E$ be an associated $*$-eigenvector. Let $G_x \times K_x$ be the gauge group defined as previously. We say that $E(x)$ is non-degenerate on $U^\alpha$ if $\forall x \in U^\alpha$ the action of $G_x \times K_x$ on $\ker(H(x)-E(x)\otimes 1_{\E})$ is transitive, i.e. if
\begin{equation}
G_x \times K_x \phi^\alpha_E(x) = \ker(H(x)-E(x)\otimes 1_{\E})
\end{equation}
\end{defi}

We see that $K_x$ has no direct influence on $\rho_E$. We introduce another gauge group $J_x$ related to $K_x$ (this relation will be apparent in the next section) but which is associated with $\rho_E$. Let $J^0_x \subset G_x$ be the union of all isotropy subgroups of the elements of the orbit by $G_x$ of $\rho_E^\alpha(x)$ :
\begin{eqnarray}
J_x^0 & = & \{j \in G_x | \exists g\in G_x, jg\rho_E^\alpha(x)g^\dagger j^\dagger = g\rho_E^\alpha(x)g^\dagger\} \nonumber \\
& = & \bigcup_{g \in G_x} G_{x,g\rho_E^\alpha(x)g^\dagger} \nonumber \\
& = & \bigcup_{g \in G_x} g G_{x,\rho_E^\alpha(x)} g^{-1}
\end{eqnarray}
($G_{x,\rho_E^\alpha(x)} = \{j \in G_x|j\rho_E^\alpha(x)j^\dagger = \rho_E^\alpha(x)\}$). $J_x^0$ is then also the orbit of the isotropy subgoup of $\rho_E^\alpha$ by the action of $G_x$ on itself. $J^0_x$ is a normal subgroup of $G_x$ (see \ref{annexeA}).\\
Suppose that $\dim \ker (\rho_E^\alpha(x)) = n-p$, let $\mathfrak j_x^1 \subset \mathfrak g_x$ be
\begin{equation}
\mathfrak j^1_x = \left\{O \in \mathfrak g_x | \hat O = \left( \begin{array}{cc} 0_{p\times p} & *_{p \times (n-p)} \\ 0_{(n-p)\times p} & *_{(n-p)\times(n-p)} \end{array} \right) \right\}
\end{equation}
where $\hat O$ is the matrix representation of $O$ in the diagonalization basis of $\rho_E^\alpha$ (the eigenvectors associated with the eigenvalue 0 being from the position $n-p$ to the position $n$), $0_{n \times p}$ is the null $n \times p$ matrix and $*_{n \times p}$ denotes any $n \times p$ matrix. $\mathfrak j^1_x$ is a normal solvable subalgebra of $\mathfrak g_x$ (see \ref{annexeA}).\\
Let $J^1_x = \{e^j, j\in \mathfrak j^1_x \}$ be the Lie group of $\mathfrak j^1_x$. $J_x$ is the direct product of groups $J_x = J^0_x \times J^1_x$. We denotes by $\mathfrak j_x = \mathfrak j^0_x \oplus \mathfrak j^1_x$ its Lie algebra, $\mathfrak j^0_x = \{j \in \mathfrak g_x|\exists g \in G, jg\rho_E^\alpha(x)g^\dagger+g \rho_E^\alpha(x)g^\dagger j^\dagger = 0 \}$. Noting that if $\rho_E^\alpha(x)$ is invertible then $J_x$ is reduced to $J^0_x$.\\

In order to have the same behaviour on the whole of $M$, we assume for the rest of this paper the following \textbf{local triviality assumption} :
\begin{itemize}
\item $\forall x,y \in M$, $G_x$ and $G_y$ are isomorphic,
\item $\forall x,y \in M$, $\dim \ker (\rho_E^\alpha(x)) = \dim \ker (\rho_E^\alpha(y))$ (and by consequences $J_x$ and $J_y$ are isomorphic).
\end{itemize}
Let $G$ be the abstract group typical of the family $\{G_x\}_{x \in M}$ and $\forall x\in U^\alpha$ let $\zeta^\alpha_x : G \to G_x \subset \mathcal{GL}(\S)$ be the group isomorphism between $G$ and $G_x$ continuous with respect to $x$ ($\zeta^\alpha_x$ is locally defined because the existence of a continuous extension on the whole of $M$ is not ensured). $\forall x,y \in U^\alpha$, $\zeta^\alpha_y \circ (\zeta^\alpha_x)^{-1}$ constitutes the isomorphism between $G_x$ and $G_y$. Let $J$ be the subgroup of $G$ such that $\zeta^\alpha_x(J) = J_x$ (by construction $J$ is the product of a normal and a solvable subgroups of $G$). The map $\zeta^\alpha_x$ defines a principal $G$-bundle on the left over $U^\alpha$ as being its fibre diffeomorphism. We denote by $P^\alpha_G$ this bundle and by $\zeta^\alpha : U^\alpha \times G \to P^\alpha_G$ its local trivialization defined by $\zeta^\alpha(x,g) = \zeta^\alpha_x(g)$. The restriction of $\zeta^\alpha$ on $U^\alpha \times J$ defines a principal $J$-bundle on the left over $U^\alpha$ denoted by $P^\alpha_J$. The bundles $P^\alpha_G$ and $P^\alpha_J$ encode the gauge theory associated with the eigenoperator equation. We denote by $P^\alpha_{\mathfrak g} = P^\alpha_G \times_G \mathfrak g$ and by $P^\alpha_{\mathfrak j} = P^\alpha_J \times_J \mathfrak j$ the vector bundles over $U^\alpha$ associated with $P^\alpha_G$ and $P^\alpha_J$ by the adjoint representation of $G$ on $\mathfrak g$, $\Ad(g) X = gXg^{-1}$ ($\mathfrak g$ and $\mathfrak j$ are the Lie algebras of $G$ and $J$). In a same manner we define the bundles $P^\alpha_K$ and $P^\alpha_{\mathfrak k}$.\\

We consider now the gauge transformations at the intersection of several charts. Let $g^{\alpha \beta} \in \mathcal C^\infty(U^\alpha \cap U^\beta) \otimes_M \bigsqcup_x G_x$ and $k^{\alpha \beta} \in \mathcal C^\infty(U^\alpha \cap U^\beta) \otimes_M \bigsqcup_x K_x$ be the 1-transition functions defined by $\forall x \in U^\alpha \cap U^\beta$
\begin{equation}
\phi^\alpha_E(x) = k^{\alpha \beta}(x) g^{\alpha \beta}(x) \phi^\beta_E(x)
\end{equation}
and then
\begin{equation}
\rho^\alpha_E(x) = g^{\alpha \beta}(x) \rho^\beta_E(x) g^{\alpha \beta}(x)^\dagger
\end{equation}
($\otimes_M$ denotes the image of the tensor product by the diagonal map which transforms $(x,y) \mapsto f(x,y)$ to $x \mapsto f(x,x)$, and $\bigsqcup$ denotes the disjoint union).\\
The 1-transition functions generate some 2-transition functions $h^{\alpha \beta \gamma} \in \mathcal C^\infty(U^\alpha \cap U^\beta \cap U^\gamma) \otimes_M \bigsqcup_x J_x$ measuring the failure of the cocycle relation (which is not ensured): $\forall x \in U^\alpha \cap U^\beta \cap U^\gamma$
\begin{equation}
h^{\alpha \beta \gamma}(x) = g^{\alpha \beta}(x) g^{\beta \gamma}(x) g^{\gamma \alpha}(x)
\end{equation}
$h^{\alpha \beta \gamma}(x)$ is an element of $J_x$ since to have a single valued mixed state, we must have $g^{\alpha \beta} g^{\beta \gamma} g^{\gamma \alpha} \rho_E^\alpha (g^{\alpha \beta} g^{\beta \gamma} g^{\gamma \alpha})^{-1} = \rho_E^\alpha$.
\begin{prop}
The 2-transition functions satisfy the following generalized cocycle relations: $\forall x \in U^\alpha \cap U^\beta \cap U^\gamma$
\begin{eqnarray}
h^{\alpha \beta \gamma}(x) & = & g^{\alpha \beta}(x) h^{\beta \gamma \alpha}(x) g^{\alpha \beta}(x)^{-1} \\
h^{\alpha \gamma \beta}(x) & = & h^{\alpha \beta \gamma}(x)^{-1} \\
h^{\beta \alpha \gamma}(x) & = & g^{\alpha \beta}(x)^{-1} h^{\alpha \beta \gamma}(x)^{-1} g^{\alpha \beta}(x)
\end{eqnarray}
and $\forall x \in U^\alpha \cap U^\beta \cap U^\gamma \cap U^\delta$
\begin{equation}
h^{\alpha \delta \gamma}(x)h^{\alpha \gamma \beta}(x) = h^{\alpha \delta \beta}(x) g^{\alpha \beta}(x) h^{\beta \delta \gamma}(x) g^{\alpha \beta}(x)^{-1}
\end{equation}
\end{prop}

\begin{proof}
\begin{equation}
h^{\alpha \beta \gamma} = g^{\alpha \beta} g^{\beta \gamma} g^{\gamma \alpha} = g^{\alpha \beta} g^{\beta \gamma} g^{\gamma \alpha} g^{\alpha \beta} g^{\beta \alpha} = g^{\alpha \beta} h^{\beta \gamma \alpha} g^{\beta \alpha} 
\end{equation}
\begin{equation}
h^{\beta \alpha \gamma} = g^{\beta \alpha} g^{\alpha \gamma} g^{\gamma \beta} = g^{\beta \alpha} g^{\alpha \gamma} g^{\gamma \beta} g^{\beta \alpha} g^{\alpha \beta} = (g^{\alpha \beta})^{-1} (h^{\alpha \beta \gamma})^{-1} g^{\alpha \beta}
\end{equation}
\begin{equation}
h^{\alpha \gamma \beta} = g^{\alpha \gamma} g^{\gamma \beta} g^{\beta \alpha} = (g^{\alpha \beta} g^{\beta \gamma} g^{\gamma \alpha})^{-1} = (h^{\alpha \beta \gamma})^{-1} 
\end{equation}
finally
\begin{equation}
g^{\alpha \delta} = h^{\alpha \delta \gamma} g^{\alpha \gamma} g^{\gamma \delta} = h^{\alpha \delta \gamma} h^{\alpha \gamma \beta} g^{\alpha \beta} g^{\beta \gamma} g^{\gamma \delta}
\end{equation}
and
\begin{equation}
g^{\alpha \delta} = h^{\alpha \delta \beta} g^{\alpha \beta} g^{\beta \delta} = h^{\alpha \delta \beta} g^{\alpha \beta} h^{\beta \delta \gamma} g^{\beta \gamma} g^{\gamma \delta}
\end{equation}
we conclude that
\begin{equation}
h^{\alpha \delta \gamma} h^{\alpha \gamma \beta} g^{\alpha \beta} = h^{\alpha \delta \beta} g^{\alpha \beta} h^{\beta \delta \gamma}
\end{equation}
\end{proof}

\section{Geometric phases of open quantum systems}
\subsection{Generator of the $C^*$-geometric phases}
Now we are able to define a generator of the geometric phases for open quantum systems by analogy with the geometric phases of dissipative quantum systems: $\mathcal A^\alpha \in \Omega^1(U^\alpha,\mathcal L(\S))$ is a solution of the following equation
\begin{equation}
\label{geomphase}
\mathcal A^\alpha \|\phi_E^\alpha\|^2_* = \langle \phi_E^\alpha |d\phi_E^\alpha \rangle_*
\end{equation}

\begin{prop}
The $C^*$-geometric phase generator and the mixed eigenstate are related by the following equation
\begin{equation}
\label{geomphasematdens}
d\rho_E^\alpha = \mathcal A^\alpha \rho_E^\alpha + \rho_E^\alpha (\mathcal A^\alpha)^\dagger
\end{equation}
\end{prop}

\begin{proof}
\begin{eqnarray}
d \rho_E & = & \tr_{\E} \left( |d\phi_E \rrangle \llangle \phi_E | + |\phi_E \rrangle \llangle d\phi_E |\right) \nonumber \\
& = & \langle \phi_E | d\phi_E \rangle_* + \langle \phi_E | d\phi_E \rangle_*^\dagger \nonumber \\
& = & \mathcal A \rho_E + (\mathcal A \rho_E)^\dagger
\end{eqnarray}
\end{proof}

\begin{prop}
\label{gaugetransf}
Under a gauge transformation $\tilde \phi_E^\alpha(x) = g(x)k(x) \phi_E^\alpha(x)$ with $g \in \Gamma(U^\alpha,P^\alpha_G)$ and $k \in \Gamma(U^\alpha,P^\alpha_K)$ the $C^*$-geometric phase generator is transformed as follows:
\begin{equation}
\tilde \mathcal A^\alpha = g\mathcal A^\alpha g^{-1} + dgg^{-1} + g \eta g^{-1}
\end{equation}
with $\eta \in \Omega^1(U^\alpha, P^\alpha_{\mathfrak j}) = \Omega^1(U^\alpha) \otimes_M \Gamma(U^\alpha,P^\alpha_{\mathfrak j})$ solution of
\begin{equation}
\eta \|\phi_E^\alpha \|^2_* = \langle \phi_E^\alpha |k^{-1}dk \phi_E^\alpha \rangle_*
\end{equation}
\end{prop}

$\Gamma(U^\alpha,P^\alpha_G)$ and $\Gamma(U^\alpha,P^\alpha_K)$ denote the sets of sections of $P^\alpha_G$ and $P^\alpha_K$.

\begin{proof}
\begin{eqnarray}
\tilde {\mathcal A} \| \tilde \phi_E \|^2_* & = & \langle \tilde \phi_E | d \tilde \phi_E \rangle_* \\
\Rightarrow \tilde {\mathcal A} g \|\phi_E\|^2_* g^\dagger & = & \tr_{\E} (k dg|\phi_E \rrangle \llangle \phi_E|g^\dagger k^{-1}) \nonumber \\
& & \quad + \tr_{\E} (g dk|\phi_E \rrangle \llangle \phi_E|k^{-1}g^\dagger) \nonumber \\
& & \quad + \tr_{\E} (gk|d\phi_E \rrangle \llangle \phi_E|k^{-1} g^\dagger) \nonumber \\
%\tilde {\mathcal A} g \|\phi_E\|^2_* g^\dagger & = & dg \|\phi_E \|^2_* g^\dagger + g \langle \phi_E | d \phi_E \rangle_* g^\dagger \nonumber \\
%& & \quad + g \langle \phi_E|k^{-1}dk|\phi_E \rangle_*g^\dagger \\
\Rightarrow \tilde {\mathcal A} g \|\phi_E\|^2_* g^\dagger & = & dg \|\phi_E \|^2_* g^\dagger + g {\mathcal A} \|\phi_E\|^2_* g^\dagger + g \eta \|\phi_E\|^2_* g^\dagger \nonumber \\
\Rightarrow g^{-1} \tilde {\mathcal A} g \|\phi_E\|^2_* & = & (g^{-1}dg + {\mathcal A}+ \eta) \|\phi_E \|^2_* \nonumber \\
%g^{-1} \tilde {\mathcal A} g & = & g^{-1}dg + {\mathcal A} + \eta \\
\Rightarrow \tilde {\mathcal A} & = & dg g^{-1} + g{\mathcal A}g^{-1} + g\eta g^{-1}
\end{eqnarray}

\begin{eqnarray}
& & \langle \phi_E|k^{-1}dk|\phi_E \rangle_*^\dagger = \langle \phi_E |d(k)^{-1}k|\phi_E \rangle_* \nonumber \\
& \Rightarrow & \langle \phi_E|k^{-1}dk|\phi_E \rangle_*^\dagger = - \langle \phi_E|k^{-1}dk|\phi_E \rangle_* \nonumber \\
%& \Rightarrow & (\eta \rho_E)^\dagger = - \eta \rho_E \\
& \Rightarrow & \eta \rho_E + \rho_E \eta^\dagger = 0 \nonumber \\
& \Rightarrow & \eta \in \Omega^1(U^\alpha,P^\alpha_{\mathfrak j_x})
\end{eqnarray}
\end{proof}
Remark : $dgg^{-1} \in \Omega^1(U^\alpha,P^\alpha_{\mathfrak g})$. We see that the $K_x$-gauge transformations on the pure state of the universe $\phi_E$ induce $\mathfrak j_x$-gauge transformations on the $C^*$-geometric phase generator. $G_x$ is the gauge group associated with the dynamics of the quantum system (the equivalent of the gauge group for closed quantum systems) whereas $J_x$ is the gauge group associated with the decoherence induced by the environment (by its relation with $K_x$ the gauge group associated with the dynamics of the environment).\\

When $\rho_E^\alpha$ is invertible the equation (\ref{geomphase}) has a single solution:
\begin{equation}
\mathcal A^\alpha = \langle \phi^\alpha_E |d\phi_E^\alpha \rangle_* (\|\phi_E^\alpha \|^2_*)^{-1}
\end{equation}
but there can be difficulty in applying this, particularly if the precise description of the universe is unknown. This problem is partially solved by the fact that we can prove (see \ref{annexeB}) that the $C^*$-geometric phase generator is equal to $\breve {\mathcal A}^\alpha = \frac{1}{2} d\rho_E^\alpha (\rho_E^\alpha)^{-1}$ modulo a $\mathfrak j_x$-gauge transformation. If $\rho_E^\alpha$ is not invertible, equation (\ref{geomphase}) can have several solutions, but they are equals modulo a $\mathfrak j_x$-gauge transformation. This result and other mathematical properties of the $C^*$-geometric phase generator can be found in   \ref{annexeB}.

\begin{prop}
We assume that $E$ is non-degenerate. The $C^*$-geometric phase generator can be decomposed as $\mathcal A^\alpha = A^\alpha + R^\alpha$ where the reduced generator $A^\alpha \in \Omega^1(U^\alpha,P^\alpha_{\mathfrak g})$ is called the gauge potential (or the $G$-potential) and where the remainder $R^\alpha$ is almost zero in the sense $\tr_{\S}(\rho_E^\alpha R^\alpha) = 0$.
\end{prop}

\begin{proof}
Let $P_E^\alpha(x) \in \mathcal L(\S \otimes \E)$ be the orthogonal projector (in sense of $\llangle .|. \rrangle$) on $\ker(H(x)-E(x)\otimes 1_{\E})$ continuous with respect to $x \in U^\alpha$. 
\begin{eqnarray}
\mathcal A \|\phi_E\|^2_* & = & \langle \phi_E|d\phi_E \rangle_* \\
& = & \langle \phi_E|P_E d\phi_E \rangle_* + \langle \phi_E|(1_{\S \otimes \E}-P_E)d\phi_E \rangle_*
\end{eqnarray}
We set then $A$ and $R$ as the solutions of the following equations
\begin{eqnarray}
A \|\phi_E \|^2_* & = & \langle \phi_E|P_E d\phi_E \rangle_* \\
R \|\phi_E \|^2_* & = & \langle \phi_E|(1-P_E)d\phi_E \rangle_*
\end{eqnarray}
By the property \ref{inner} we have $\langle \phi_E^\alpha|P_E^\alpha d\phi_E^\alpha \rangle_* \in \Omega^1(U^\alpha,P^\alpha_{\mathfrak g})$ and then $A^\alpha \in \Omega^1(U^\alpha,P^\alpha_{\mathfrak g})$.
\begin{eqnarray}
\tr_{\S}(\rho_E R) & = & \tr_{\S}\langle \phi_E|(1-P_E)d\phi_E \rangle_* \nonumber \\
& = & \llangle \phi_E|(1-P_E) d\phi_E \rrangle \nonumber \\
& = & \llangle (1-P_E) \phi_E|d\phi_E \rrangle \nonumber \\
& = & 0
\end{eqnarray}
\end{proof}

\begin{prop}
Under a gauge transformation $\tilde \phi_E^\alpha(x) = g(x)k(x) \phi_E^\alpha(x)$ with $g \in \Gamma(U^\alpha,P^\alpha_G)$ and $k \in \Gamma(U^\alpha,P^\alpha_K)$ the $G$-potential is transformed as follows:
\begin{equation}
\tilde A^\alpha = g A^\alpha g^{-1} + dgg^{-1} + g \eta g^{-1}
\end{equation}
with $\eta \in \Omega^1(U^\alpha, P^\alpha_{\mathfrak j})$ solution of
\begin{equation}
\eta \|\phi_E^\alpha \|^2_* = \langle \phi_E^\alpha |k^{-1}dk \phi_E^\alpha \rangle_*
\end{equation}
and the remainder is transformed as follows :
\begin{equation}
\tilde R^\alpha = gR^\alpha g^{-1}
\end{equation}
\end{prop}

\begin{proof}
By definition, $gP_E = P_Eg$ and $kP_E = P_Ek$, and since $dgg^{-1} \in \Omega^1(U^\alpha,P^\alpha_{\mathfrak g})$ and $dkk^{-1} \in \Omega^1(U^\alpha,P^\alpha_{\mathfrak k})$ we have $(1_{\S \otimes \E}-P_E) dgg^{-1} = 0$ and $(1_{\S \otimes \E} - P_E)dkk^{-1}=0$. The rest of the proof consists to remake the proof of property \ref{gaugetransf} by considering these identities.
\end{proof}

We consider now the gauge transformations at the intersection of several charts. Let $g^{\alpha \beta} \in \mathcal C^\infty(U^\alpha \cap U^\beta) \otimes_M \bigsqcup_x G_x$ and $k^{\alpha \beta} \in \mathcal C^\infty(U^\alpha \cap U^\beta) \otimes_M \bigsqcup_x K_x$ be the 1-transition functions. We have clearly $\forall x \in U^\alpha \cap U^\beta$
\begin{equation}
A^\beta = (g^{\alpha \beta})^{-1} A^\alpha g^{\alpha \beta} - (g^{\alpha \beta})^{-1} dg^{\alpha \beta} + (g^{\alpha \beta})^{-1} \eta^{\alpha \beta} g^{\alpha \beta}
\end{equation}
where $\eta^{\alpha \beta} \in \Omega^1(U^\alpha \cap U^\beta) \otimes \bigsqcup_x \mathfrak j_x$ is a solution of the equation:
\begin{equation}
\eta^{\alpha \beta} \|\phi^\alpha_E\|^2_* = \langle \phi^\alpha_E|k^{\alpha \beta}d(k^{\alpha \beta})^{-1}|\phi^\alpha_E \rangle_*
\end{equation}
$\eta^{\alpha \beta}$ is called the potential-transformation (or the $J$-potential) and is associated with the decoherence induced by the environment (since this term vanishes for a closed quantum system).

\begin{prop}
$\forall x \in U^\alpha \cap U^\beta \cap U^\gamma$, the potential-transformation satisfies
\begin{eqnarray}
& & \eta^{\alpha \beta} + g^{\alpha \beta} \eta^{\beta \gamma} (g^{\alpha \beta})^{-1} - h^{\alpha \beta \gamma} \eta^{\alpha \gamma} (h^{\alpha \beta \gamma})^{-1} \nonumber \\
& & \quad = dh^{\alpha \beta \gamma}(h^{\alpha \beta \gamma})^{-1} - [A^\alpha,h^{\alpha \beta \gamma}](h^{\alpha \beta \gamma})^{-1}
\end{eqnarray}
where $h^{\alpha \beta \gamma} \in \mathcal C^\infty(U^\alpha \cap U^\beta \cap U^\gamma) \otimes_M \bigsqcup_x J_x$ is the 2-transition functions.
\end{prop}

\begin{proof}
\begin{eqnarray}
\eta^{\alpha \beta} & = & g^{\alpha \beta} A^\beta (g^{\alpha \beta})^{-1} - A^\alpha + dg^{\alpha \beta} (g^{\alpha \beta})^{-1} \nonumber \\
& = & g^{\alpha \beta} \left(g^{\beta \gamma} A^\gamma (g^{\beta \gamma})^{-1} +dg^{\beta \gamma}(g^{\beta \gamma})^{-1} - \eta^{\beta \gamma}\right)(g^{\alpha \beta})^{-1} \nonumber \\
& & \quad - A^\alpha + dg^{\alpha \beta} (g^{\alpha \beta})^{-1} \nonumber \\
%& = & g^{\alpha \beta} \left(g^{\beta \gamma} \left(g^{\gamma \alpha} A^\alpha (g^{\gamma \alpha})^{-1} + dg^{\gamma \alpha}(g^{\gamma \alpha})^{-1} - \eta^{\gamma \alpha} \right) (g^{\beta \gamma})^{-1} \right. \nonumber \\
%& & \quad \left. +dg^{\beta \gamma}(g^{\beta \gamma})^{-1} - \eta^{\beta \gamma}\right)(g^{\alpha \beta})^{-1} - A^\alpha + dg^{\alpha \beta} (g^{\alpha \beta})^{-1} \\
& = & h^{\alpha \beta \gamma} A^\alpha (h^{\alpha \beta \gamma})^{-1} - A^\alpha - g^{\alpha \beta} \eta^{\beta \gamma} (g^{\alpha \beta})^{-1} \nonumber \\
& & \quad - g^{\alpha \beta}g^{\beta \gamma} \eta^{\gamma \alpha} (g^{\alpha \beta}g^{\beta \gamma})^{-1} \nonumber \\
& & \quad \left. \begin{array}{l} + g^{\alpha \beta}g^{\beta \gamma} dg^{\gamma \alpha} (g^{\alpha \beta}g^{\beta \gamma}g^{\gamma \alpha})^{-1} \\ +g^{\alpha \beta}dg^{\beta \gamma}(g^{\alpha \beta}g^{\beta \gamma})^{-1} \\ + dg^{\alpha \beta} (g^{\alpha \beta})^{-1} \end{array} \right\}dh^{\alpha \beta \gamma}(h^{\alpha \beta \gamma})^{-1}
\end{eqnarray} 
now
\begin{eqnarray}
\eta^{\gamma \alpha} & = & g^{\gamma \alpha} A^\alpha (g^{\gamma \alpha})^{-1} +dg^{\gamma \alpha}(g^{\gamma \alpha})^{-1} - A^\gamma \nonumber \\
%& = & g^{\gamma \alpha} \left(g^{\alpha\gamma} A^\gamma (g^{\alpha \gamma})^{-1}+dg^{\alpha\gamma} (g^{\alpha \gamma})^{-1}-\eta^{\alpha \gamma} \right)(g^{\gamma \alpha})^{-1} \nonumber \\
%& & \quad  +dg^{\gamma \alpha}(g^{\gamma \alpha})^{-1} - A^\gamma  \\
& = & A^\gamma + g^{\gamma \alpha}d(g^{\gamma \alpha})^{-1} - g^{\gamma \alpha} \eta^{\alpha \gamma} (g^{\alpha \gamma})^{-1} +dg^{\gamma \alpha}(g^{\gamma \alpha})^{-1} - A^\gamma \nonumber \\
& = & - g^{\gamma \alpha} \eta^{\alpha \gamma} (g^{\alpha \gamma})^{-1}
\end{eqnarray}
then
\begin{eqnarray}
\eta^{\alpha \beta} & = & h^{\alpha \beta \gamma} A^\alpha (h^{\alpha \beta \gamma})^{-1} - A^\alpha - g^{\alpha \beta} \eta^{\beta \gamma} (g^{\alpha \beta})^{-1} \nonumber \\
& & \quad + g^{\alpha \beta} g^{\beta \gamma} g^{\gamma \alpha} \eta^{\alpha \gamma} (g^{\alpha \beta} g^{\beta \gamma} g^{\gamma \alpha})^{-1} +dh^{\alpha \beta \gamma}(h^{\alpha \beta \gamma})^{-1}
\end{eqnarray}
\end{proof}

\subsection{The higher gauge theory associated with the $C^*$-geometric phase }
The gauge transformations of the gauge potential $A^\alpha$ and of the potential-transformation $\eta^{\alpha \beta}$ are characteristic of a higher gauge theory, as described with several different notations in refs. \cite{breen,aschieri,kalkkinen,baez1,baez2,baez3}. In accordance with these previous works, we introduce here three types of curvature.\\
Let $B^\alpha = d\mathcal A^\alpha - \mathcal A^\alpha \wedge \mathcal A^\alpha \in \Omega^2(U^\alpha,\mathcal L(\S))$ be the curving and $F^\alpha = dR^\alpha - [A^\alpha,R^\alpha] - R^\alpha \wedge R^\alpha \in \Omega^2(U^\alpha,\mathcal L(\S))$ be the fake curvature. 
\begin{equation}
B^\alpha = dA^\alpha - A^\alpha \wedge A^\alpha + F^\alpha
\end{equation}
We can prove (see   \ref{annexeC}) that in fact $B^\alpha \in \Omega^2(U^\alpha,P^\alpha_{\mathfrak j})$ and $F^\alpha \in \Omega^2(U^\alpha,P^\alpha_{\mathfrak g})$. Finally we set $H^\alpha \in \Omega^3(U^\alpha,P^\alpha_{\mathfrak j})$ the true curvature defined by
\begin{eqnarray}
H^\alpha & = & dB^\alpha - [A^\alpha,B^\alpha] \nonumber \\
& = & dF^\alpha - [A^\alpha,F^\alpha]
\end{eqnarray}
These three types of curvature satisfy some gauge transformation formulae (studied in   \ref{annexeC}) which are also characteristic of a higher gauge theory \cite{breen,aschieri,kalkkinen,baez1,baez2,baez3}).\\
The $C^*$-geometric phases of an open quantum system are not related to a principal bundle over $M$ (as would occur for the geometric phases of closed quantum systems) because the family of local principal $G$-bundles $\{P^\alpha_G\}_{\alpha}$ cannot be lift to a single global principal bundle (this is due to the failure of the cocycle relation with $g^{\alpha \beta}$ measured by $h^{\alpha \beta \gamma}$). The geometric structure $\mathcal P$ associated with the family $\{P^\alpha_G\}_{\alpha}$ does not define a manifold but a category (see ref. \cite{maclane} for a presentation of category theory); this feature is specific to a higher gauge theory.\\
The geometric structure $\mathcal P$ is defined from the 1-transition $g^{\alpha \beta}$ and the 2-transition $h^{\alpha \beta \gamma}$ and is endowed with a 2-connection (see ref. \cite{baez1,baez2,baez2}) defined from $(A^\alpha,B^\alpha,\eta^{\alpha \beta})$. $\mathcal P$ is a non-abelian bundle gerbes \cite{breen,aschieri,kalkkinen,baez1,baez2,baez3})\footnote[8]{In some references as \cite{breen,aschieri,kalkkinen}, the authors consider only the special crossed module $(\mathrm{Aut}(G),G,\Ad,id_{\mathrm{Aut}(G)})$ for a Lie group $G$.} with the structure Lie crossed module $(G,J,t,\Ad)$ where $t:J \to G$ is the canonical injection ($J$ is a subgroup of $G$) and $\Ad : G \to \mathrm{Aut}(J)$ is the adjoint representation of $G$ on itself restricted to $J$ (in the homomorphisms domain).\\
The replacement of the single gauge group of  closed quantum systems by a gauge Lie crossed module for open quantum systems is explained by the need for a gauge group associated with the evolution of the quantum system ($G$) and also for a gauge group associated with the decoherence induced by the environment ($J$).\\ 

Equivalently, $\mathcal P$ can be viewed as a principal categorical bundle (a 2-bundle, see ref. \cite{baez1,baez2,baez3}) on the left with the structure groupo\"id $\mathcal G$ having $\Obj(\mathcal G) = G$ as set of objects and $\Morph(\mathcal G) = J \rtimes G$ as set of arrows, the semi-direct product (the arrows horizontal composition) being defined by $(h,g)(h',g') = (h\Ad(g)h',gg')$. The source map of $\mathcal G$ is defined by $s(h,g) = g$ and the target map is defined by $t(h,g) = t(h)g$ (with $t(h)$ the canonical injection of $h$ in $G$). The morphisms composition (the arrows vertical composition) is defined by $(h,g) \circ (h',t(h)g) = (hh',g)$.\\

The explicit constructions of $\mathcal P$, of its 2-connection and of its total category is somewhat technical but can be found in \ref{annexeC}. In the next section, we show that at the adiabatic limit, the time-dependent mixed states exhibit $C^*$-geometric phases which are related to horizontal lifts in $\mathcal P$.

\section{The adiabatic transport of mixed states}
The study of rigorous adiabatic theorems for mixed states is not the subject of this paper but   \ref{annexeD} presents an heuristic approach to the adiabatic approximation for the $*$-eigenvectors. This approach shows that the adiabatic condition can be expressed as being $R^\alpha_\mu(x(t)) \dot x^\mu(t) \simeq 0$ ($\forall t$) where $R^\alpha$ is the remainder of $C^*$-geometric phase generator. This property shows the interpretation of $R^\alpha$, it measures the non-adiabatic effects. The fake curvature $F^\alpha = dR^\alpha - R^\alpha \wedge R^\alpha - [A^\alpha,R^\alpha]$ then characterizes the influence of these non-adiabatic effects on the geometry of $\mathcal P$ and the curving-transformation $\chi^{\alpha \beta} = [R^\alpha,\eta^{\alpha \beta}]$ measures the intertwining between the non-adiabatic effects and the decoherence induced by the environment.\\

We now are able to introduce the adiabatic transport of a mixed state.
\begin{prop}
\label{propadiabtransp}
  Let $(\mathcal D(\S),\mathcal L_x,M)$ be an open quantum dynamical system associated with an universe described by the quantum dynamical system $(\S\otimes \E,H(x),M)$. Let $U^\alpha \ni x \mapsto \rho_E^\alpha(x) \in \mathcal D(\S)$ be a mixed eigenstate of the open quantum system and $M \ni x \mapsto E(x)$ be the associated eigenoperator supposed \textbf{non-degenerate}. Let $t \mapsto x(t) \in M$ be an evolution corresponding to a path $\mathcal C$ within $U^\alpha$. Let $t \mapsto \rho(t) \in \mathcal D(\S)$ be the solution of the Lindblad equation $\ihbar \frac{d\rho}{dt} = \mathcal L_{x(t)}(\rho(t))$ with $\rho(0) = \rho_E^\alpha(x(0))$. We assume the adiabatic condition (\ref{adiabcond}), i.e. $R^\alpha_\mu(x(t))\dot x^\mu(t) \simeq 0$. We can then write $\forall t>0$
\begin{equation}
\label{adiabtransp}
\rho(t) = g_{EA}(t) \rho_E(x(t)) g_{EA}(t)^\dagger
\end{equation}
with
\begin{equation}
g_{EA}(t) = \Te^{- \ihbar^{-1} \int_0^t E(x(t'))dt'} \Pe^{-\int_{x(0)}^{x(t)}(A^\alpha(x)+\eta(x))} \in G_{x(t)}
\end{equation}
where $A^\alpha \in \Omega^1(U^\alpha,P^\alpha_{\mathfrak g})$ is the gauge potential and $\eta \in \Omega^1(\mathcal C,P^\alpha_{\mathfrak j})$ is a particular $\mathfrak j_x$-gauge transformation.
\end{prop}

\begin{proof}
Let $\psi \in \S \otimes \E$ be the solution of the Schr\"odinger equation $\ihbar \frac{d\psi}{dt} = H(x(t)) \psi(t)$ with $\psi(0) = \phi_E(x(0))$ where $\phi_E$ is the $*$-eigenvector associated with $E$ and $\rho_E$. By application of the theorem \ref{thadiab} we know that $\forall t$, $\psi(t) \in \ker(H(x(t))-E(x(t)))$, and since the action of $G_x \times K_x$ is transitive ($E$ is non-degenerate) we have
\begin{equation}
\forall t, \exists g_{EA}(t) \in G_{x(t)}, \exists k_\eta(t) \in K_{x(t)}, \psi(t) = k_\eta(t)g_{EA}(t) \phi_E(x(t))
\end{equation}
By inserting this expression in the Schr\"odinger equation of the universe, we find
\begin{eqnarray}
\ihbar k_{\eta} \dot g_{EA} \phi_E + \ihbar \dot k_{\eta} g_{EA} \phi_E + \ihbar k_{\eta} g_{EA} \frac{d \phi_E}{dt} & = & H k_{\eta} g_{EA} \phi_E \nonumber \\ & = & E k_{\eta} g_{EA} \phi_E 
\end{eqnarray}
\begin{equation}
g_{EA}^{-1} \dot g_{EA} \phi_E = - \frac{d \phi_E}{dt} - k_{\eta}^{-1} \dot k_{\eta} \phi_E  - \ihbar^{-1} g_{EA}^{-1} E g_{EA} \phi_E
\end{equation}
and then
\begin{eqnarray}
 g_{EA}^{-1} \dot g_{EA} |\phi_E \rrangle \llangle \phi_E| & = & - | \frac{d \phi_E}{dt} \rrangle \llangle \phi_E| - k_{\eta}^{-1} \dot k_{\eta}|\phi_E \rrangle \llangle \phi_E| \nonumber \\
& & \quad - \ihbar^{-1} g_{EA}^{-1} E g_{EA} |\phi_E \rrangle \llangle \phi_E |
\end{eqnarray}
By taking the partial trace on $\E$ of this expression, we find
\begin{equation}
 g_{EA}^{-1} \dot g_{EA} \rho_E = - (\mathcal A_\mu+ \eta_\mu)\dot x^\mu \rho_E - \ihbar^{-1} g_{EA}^{-1} E g_{EA} \rho_E
\end{equation}
where the potential-transformation $\eta$ is a solution of the following equation
\begin{eqnarray}
& & \eta(x(t)) \|\phi_E(x(t))\|^2_* \nonumber \\
& & \qquad = \langle \phi_E(x(t)) | k_\eta(x(t))^{-1} \frac{dk_{\eta}(x(t))}{dt} \phi_E(x(t)) \rangle_* dt 
\end{eqnarray}
Since $R_\mu \dot x^\mu \simeq 0$, $\mathcal A_\mu \dot x^\mu \simeq A_\mu \dot x^\mu$ we can then choose
\begin{equation}
g_{EA}^{-1} \dot g_{EA} = -(A_\mu+\eta_\mu)\dot x^\mu - \ihbar^{-1} g_{EA}^{-1} E g_{EA}
\end{equation}
We set
\begin{eqnarray}
& & g_{A}(t)  =  \Te^{\ihbar^{-1} \int_0^t E(t)dt} g_{EA}(t)  \nonumber \\
& \iff & g_{EA}(t) = \Te^{-\ihbar^{-1} \int_0^t E(t)dt} g_{A}(t)
\end{eqnarray}
By inserting this expression in the previous equation we find
\begin{equation}
 g_{A}^{-1} \dot g_{A} = - (A_\mu+\eta_\mu)\dot x^\mu
\end{equation}
and thus
\begin{equation}
g_A = \Pe^{- \int_{x(0)}^{x(t)} (A(x)+\eta(x))}
\end{equation}
\end{proof}
We note that even if the $C^*$-geometric phase $ \Pe^{\int_{x(0)}^{x(t)}(A^\alpha(x)+\eta(x))}$ is non-abelian since it belongs to $\mathcal{GL}(\S)$ it is the equivalent of an abelian geometric phase of a closed quantum system. The non-abelianity results from the substitution of the abelian ring $\mathbb C$ by the non-abelian $C^*$-algebra $\mathcal L(\S)$ (the paradigm of the present approach).\\
$A^\alpha$ being not antisefladjoint, $g_{EA}$ is not unitary, nevertheless the normalization of the density matrix trace is preserved by the adiabatic transport formula (see   \ref{annexeD}). If the path $\mathcal C$ crosses several charts, the $C^*$-geometric phase is obtained by a formula similar to equation (\ref{alvaformule}).

\section{Illustrative example: the control of a qubit}
In this section we present a very simple example in order to illustrate the concepts introduced in this paper with a physical toy model. For the sake of simplicity, we do not consider the question of the chart indices in this section.
\subsection{The model}
We consider a qubit described by a CAR algebra (a fermionic algebra) acting on $\S = \mathbb C^2$:
\begin{eqnarray}
c|0 \rangle = 0 & \qquad & c^+|0\rangle = |1 \rangle \nonumber \\
c|1 \rangle = |0 \rangle & \qquad & c^+|1 \rangle = 0
\end{eqnarray}
with $cc^+ + c^+c = 1_{\S}$. This qubit is subjected to a decoherence process described by a phase damping model (see \cite{nielsen}):
\begin{equation}
H_{I0} = \chi c^+c \otimes (b+b^+)
\end{equation}
where $\chi \in \mathbb R$ is a constant. Usually $\{b,b^+,1_{\E}\}$ constitutes a CCR algebra (a harmonic oscillator algebra), but in the present work for the sake of simplicity and to avoid some technical difficulties we consider that it is also a CAR algebra (the environment is then constituted by a single mode fermionic bath). We have then the following hamiltonian:
\begin{equation}
H_0 = \hbar \omega_c c^+c \otimes 1_{\E} + \hbar \omega_b 1_{\S} \otimes b^+b + \chi c^+c \otimes (b+b^+)
\end{equation}
where $\omega_c,\omega_b \in \mathbb R^{+*}$ are some constants. We control the qubit by operating on it with rotations. Let $U(x) = e^{\imath x^\mu \sigma_\mu} \in \mathcal U(\S)$ be a qubit rotation, where $x=(x^0,x^1,x^3) \in \mathbb R^3$ are the control parameters and $\{\sigma_\mu \}_{\mu=1,2,3}$ are the Pauli matrices. The Hamiltonian of the qubit driven by rotations is then
\begin{eqnarray}
H(x) & = & U(x) H_0 U(x)^{-1} \nonumber \\
& = & \hbar \omega_c U(x) c^+c U(x)^{-1} \otimes 1_{\E} + \hbar \omega_b 1_{\S} \otimes b^+b \nonumber \\
& & \quad + \chi U(x) c^+c U(x)^{-1} \otimes (b+b^+)
\end{eqnarray}
Let $t \mapsto x(t)$ be a path in $\mathbb R^3$ with $x(0)=(0,0,0)$ and let $t \mapsto \psi(t) \in \S \otimes \E$ be the associated wave function of the universe:
\begin{equation}
\ihbar \frac{d \psi}{dt} = U(x(t)) H_0 U(x(t))^{-1} \psi(t) \qquad \psi(0) = \phi_\chi
\end{equation}
with the initial condition $\phi_\chi \in \S \otimes \E$ such that $\rho_\chi = \tr_{\E}\left(|\phi_\chi \rrangle \llangle \phi_\chi|\right) = |1 \rangle \langle 1|$ (at time $t=0$ the qubit is in the pure state $|1\rangle$). Let $\phi(t) = U(x(t))^{-1} \psi(t)$ be the solution of
\begin{equation}
\ihbar \frac{d\phi}{dt} = \left(H_0 - \ihbar U(x(t))^{-1} \frac{dU(x(t))}{dt} \right) \phi(t) \qquad \phi(0)=\phi_\chi
\end{equation}
The density matrix representing the evolution of the qubit is then
\begin{equation}
\rho(t) = U(x(t)) \tr_{\E} \left( V(t) |\phi_\chi \rrangle \llangle \phi_\chi| V(t)^\dagger \right) U(x(t))^{-1}
\end{equation}
with
\begin{equation}
V(t) = \Te^{-\ihbar^{-1} \int_0^t (H_0 - \ihbar U(x(t'))^{-1} \partial_{t'} U(x(t'))) dt'} \in \mathcal U(\S \otimes \E)
\end{equation}
We see that the dynamics of the qubit seems to exhibit a kind of geometric phase generated by $U(x)^{-1} dU(x)$. Unfortunately it is not separated from the other matrix quantities and the expression does not involve the rotated state $U(x)|1\rangle \langle 1|U(x)^{-1}$. The formalism introduced in this paper solves these problems and clearly defines this geometric phase.

\subsection{The eigenoperator and the $*$-eigenvector}
Let $E_0 = \hbar \omega_c cc^+ \in \mathcal L(\S)$. We have the following usual eigenvector equation
\begin{equation}
(H_0 - E_0 \otimes 1_{\E}) \phi_\chi = \lambda \phi_\chi
\end{equation}
with $\lambda = \frac{1}{2} (\sqrt{4 \chi^2+\hbar^2 \omega_b^2}+2\hbar \omega_c+\hbar \omega_b)$ (one of the four eigenvalues) and
\begin{equation}
\phi_\chi = \frac{2 \chi |10 \rangle + (\sqrt{4 \chi^2+\hbar^2 \omega_b^2}+\hbar \omega_b) |11 \rangle }{\sqrt{4 \chi^2+(\sqrt{4 \chi^2+\hbar^2 \omega_b^2}+\hbar \omega_b)^2}} 
\end{equation}
We note that $\rho_\chi = |1 \rangle \langle 1|$.\\
We consider the following eigenoperator and $*$-eigenvector of $H(x)$ defined by
\begin{equation}
E(x) = U(x)(E_0+\lambda 1_{\S}) U(x)^{-1} \qquad \phi_E(x) = U(x) \phi_\chi
\end{equation}
It is easy to verify that
\begin{equation}
H(x) \phi_E(x) = E(x) \phi_E(x) \text{ and } [H(x),E(x)]=0
\end{equation}
The mixed eigenstate is then $\rho_E(x) = U(x) |1\rangle \langle 1| U(x)^{-1}$ which is as expected the rotated initial state. The associated gauge invariances are then described by the following Lie algebras
\begin{equation}
\mathfrak g_x = \left\{ U(x) \left( \begin{array}{cc} \alpha & 0 \\ \gamma & \beta \end{array} \right) U(x)^{-1}; \alpha,\gamma \in \mathbb C, \beta \in \mathbb C \right\}
\end{equation}
\begin{equation}
\mathfrak j_x^0 = \left\{ U(x) \left( \begin{array}{cc} \alpha & 0 \\ \gamma & \beta \end{array} \right) U(x)^{-1}; \alpha,\gamma \in \mathbb C, \beta \in \imath \mathbb R \right\}
\end{equation}
\begin{equation}
\mathfrak j_x^1 = \left\{ U(x) \left( \begin{array}{cc} \alpha & 0 \\ \gamma & 0 \end{array} \right) U(x)^{-1}; \alpha,\gamma \in \mathbb C \right\}
\end{equation}
and $\mathfrak k_x = \imath \mathbb R$.

\subsection{The $C^*$-geometric phase}
By definition the generator of the $C^*$-geometric phase $\mathcal A$ satisfies
\begin{eqnarray}
& & \mathcal A \rho_E = \langle \phi_E|d\phi_E \rangle_* \nonumber \\
& \iff & \mathcal A U(x) |1 \rangle \langle 1| U(x)^{-1} =  dU(x) |1 \rangle \langle 1| U(x)^{-1} \nonumber \\
& \iff & U(x)^{-1} \mathcal A U(x) |1 \rangle \langle 1| = U(x)^{-1} dU(x) |1 \rangle \langle 1|
\end{eqnarray}
This induces that
\begin{equation}
\mathcal A(x) = U(x) \left( \begin{array}{cc} 0 & \langle 0|U(x)^{-1}dU(x) |1 \rangle \\ 0 & \langle 1|U(x)^{-1}dU(x)|1 \rangle \end{array} \right) U(x)^{-1}
\end{equation}
modulo a $\mathfrak j_x^1$-gauge transformation. $\mathcal A$ exhibits the relevant components of the generator $U^{-1}(x)dU(x)$ postulated at the begining of this section. Moreover we have
\begin{eqnarray}
\langle \phi_E |P_E d\phi_E \rangle_* & = & \llangle \phi_\chi | U(x)^{-1}dU(x) |\phi_\chi \rrangle U(x)|1 \rangle \langle 1|U(x)^{-1} \nonumber \\
& = & \langle 1|U(x)^{-1}dU(x)|1 \rangle U(x)|1 \rangle \langle 1|U(x)^{-1}
\end{eqnarray}
where $P_E = |\phi_E \rrangle \llangle \phi_E|$. This shows that the reduced geometric phase generator is
\begin{equation}
A(x) = U(x) \left( \begin{array}{cc} 0 & 0 \\ 0 & \langle 1|U(x)^{-1}dU(x)|1 \rangle \end{array} \right) U(x)^{-1}
\end{equation}
The adiabatic condition $R_\mu(x(t))\dot x^\mu(t) \simeq 0$ is then equivalent to 
\begin{equation}
\langle 0|U(x)^{-1} \frac{\partial U(x)}{\partial x^\mu} |1 \rangle \dot x^\mu(t) \simeq 0
\end{equation}
The non-adiabatic coupling between $|1 \rangle$ and $|0 \rangle$ induced by the rotation must then be negligible; this is in accordance with the intuitive signifiance of the adiabatic approximation.\\
Finally, by applying the property \ref{propadiabtransp} we have
\begin{equation}
\rho(t) = g_{EA}(t) U(x(t)) |1 \rangle \langle 1| U(x(t))^{-1} g_{EA}(t)^{-1}
\end{equation}
with
\begin{eqnarray}
g_{EA}(t) & = & \Te^{-\ihbar^{-1} \int_0^t U(x(t'))(E_0+\lambda)U(x(t'))^{-1}dt'} \nonumber \\
& & \times  \Pe^{- \int_{x(0)}^{x(t)} \langle 1|U(x)^{-1}dU(x)|1\rangle U(x)|1 \rangle \langle 1|U(x)^{-1}}
\end{eqnarray}
The true geometric phase associated with the adiabatic rotation of the qubit subjected to the decoherence process explicitly appears in this formula.

\section{Conclusion}
The analogy between dissipative and open quantum systems, summarized by table \ref{analogy}, induces a new approach to extend the geometric phase concept to open quantum systems.
\begin{table}
\caption{\label{analogy} Analogy between dissipative and open quantum systems.}
\begin{center}
\begin{tabular}{|c|c|}
\hline
\textit{Dissipative quantum systems} & \textit{Open quantum systems} \\
\hline
ring $\mathbb C$ & $C^*$-algebra $\mathcal L(\S)$ \\
\hline
Hilbert space $\mathbb C \otimes \mathcal H = \mathcal H$ & $C^*$-module $\S \otimes \E$ \\
\hline
dissipation density $\|\psi\|^2 \in \mathbb R^{*+}$ & density matrix $\|\psi\|^2_* \in \mathcal D(\S)$ \\
\hline
$ \ihbar \frac{d\|\psi\|^2}{dt} = \mathcal L(\|\psi\|^2)$ & $\ihbar \frac{d\|\psi\|^2_*}{dt} = \mathcal L(\|\psi\|^2_*) $\\
$\mathcal L(\|\psi\|^2) = \langle \psi|(H-H^\dagger)\psi\rangle$ & $\mathcal L(\|\psi\|^2_*) = \langle \psi|H\psi\rangle_* - \langle H\psi|\psi\rangle_*$ \\
\hline
$H \phi_e = e \phi$ with $e\in \mathbb C$ & $H\phi_E=E\phi_E$ with $E\in \mathcal L(\S)$ \\
$\mathcal L(\|\phi_e\|^2) = 2 \imath \im(e)\|\phi_e\|^2$ & $\mathcal L(\|\phi_E\|^2_*) = E\|\phi_E\|^2_* - \|\phi_E\|^2_*E^\dagger$ \\
\hline
$A = \frac{\langle \phi_e|d\phi_e \rangle}{\|\phi_e\|^2}$ & $\mathcal A\|\phi_E\|^2_* = \langle \phi_E|d\phi_E \rangle_*$ \\
$\frac{d\|\phi_e\|^2}{\|\phi_e\|^2} = 2 \re(A)$ & $d\|\phi_E\|^2_* = \mathcal A\|\phi_E\|^2_* + \|\phi_E\|^2_* \mathcal A^\dagger$ \\
\hline
\end{tabular}
\end{center}
\end{table}
This new kind of geometric phase defines a higher gauge theory (and not the usual gauge theory) based on the Lie crossed module $(G,J,t,\Ad)$, where the group $G$ is the analogue of the usual gauge group of closed quantum systems whereas the group $J$ is related to the decoherence induced by the environment on the open quantum system ($J$ is reduced to $\{1_G\}$ for a closed system). The $C^*$-geometric phases then appear in a categorical generalization of a principal bundle defined by 1-transition functions $g^{\alpha \beta}$ associated with $G$ but for which the cocycle relation fails, the failure being measured by 2-transition functions $h^{\alpha \beta \gamma}$ associated with $J$. We therefore claim that there is an equivalence between the three following phenomena:
\begin{itemize}
\item \textbf{Dynamics:} decoherence induced by the environment.
\item \textbf{Geometry:} obstruction to lift the gauge structure onto a principal bundle.
\item \textbf{Topology:} failure of the cocycle relation.
\end{itemize}
The categorical bundle is endowed with a connection described by several data summarised in table \ref{connectiondata} including the generator of the $C^*$-geometric phases.
\begin{table}
\caption{\label{connectiondata} Connection data of the higher gauge theory associated with the $C^*$-geometric phases.}
\begin{center}
\begin{tabular}{|c|c|c|c|l|}
\hline
\textit{$n$-form} & \textit{Symbol} & \textit{Degree} & \textit{Values} & \textit{Interpretation} \\
              &                 & $n$       & \textit{algebra}& \\
\hline
gauge potential & $A^\alpha$ & 1 & $\mathfrak g_x$ & characterizes the adiabatic dynamics \\
\hline
remainder & $R^\alpha$ & 1 & $\mathcal L(\S)$ & measures the non-adiabatic effects \\
\hline
potential-transformation & $\eta^{\alpha \beta}$ & 1 & $\mathfrak j_x$ & manifestation of the decoherence \\
\hline
curving & $B^\alpha$ & 2 & $\mathfrak j_x$ & characterizes the decoherence geometry \\
\hline
fake curvature & $F^\alpha$ & 2 & $\mathfrak g_x$ & characterizes the non-adiabatic geometry \\
\hline
%adiabatic curvature & $dA^\alpha - A^\alpha \wedge A^\alpha$ & 2 & $\mathfrak g_x$ & characterizes the adiabatic geometry \\
%\hline
curving-transformation & $\chi^{\alpha \beta}$ & 2 & $\mathfrak j_x$ & nonadiabaticy-decoherence intertwining \\
\hline
true curvature & $H^\alpha$ & 3 & $\mathfrak j_x$ & characterizes the geometry \\
\hline
\end{tabular}
\end{center}
\end{table}

The relation between decoherence and a higher gauge theory seems to be important. In previous works we have already pointed out that a higher gauge theory arises for some geometric phases: in ref. \cite{viennot2} for a dissipative quantum system presenting a splitting resonance crossing (the decoherence is there induced by the interplay of the dissipative process and the non-adiabatic transitions induced by the crossing), and in ref. \cite{viennot3} for an atom or a molecule interacting with a modulated chirped laser field or with an irregular train of ultrashort pulses (the decoherence is there due to the photon exchanges between the quantum system and the laser field or to the chaotic behaviour of the kicked quantum system). In these previous works, the quantum systems are described by pure states, and the associated higher gauge theories are abelian. In the present work, the open quantum systems are described by mixed states and the higher gauge theory is non-abelian.\\

The usual gauge theories were initially introduced to describe particle interactions (electromagnetic, electroweak or chromodynamics), and were later used to describe the geometric phases of closed quantum systems. The higher gauge theories used in the present work to describe the geometric phases of open quantum systems were initially introduced to describe string interactions. We note the interesting epistemological remark that, from the viewpoint of the gauge theories, the reduction to a more fundamental theory in high energy physics (the passage from the quantum field theory to the string theory) seems to be equivalent in non-relativistic quantum dynamics to an increase of the complexity (the passage from a closed quantum system to a system submitted to a decoherence process).\\

The $C^*$-geometric phases presented in this paper are associated with horizontal lifts of curves in the categorical bundle. However we can note that the higher gauge theories also permit us to define horizontal lifts of surfaces. The role of such horizontal lifts of surfaces in the context of open quantum systems is at the present time unclear. Moreover for some cases, $J$ can be reduced to $U(1)$. In these cases the generalized cocycle relation $h^{\alpha \beta \gamma} h^{\alpha \gamma \beta} = h^{\alpha \delta \beta} h^{\beta \delta \gamma}$ induces the existence of a topological invariant, the Dixmier-Douady class which measures the non-triviality of the categorical bundle. The physical interpretation of this class is also at the present time unclear.

\ack
This work is supported by grants from \textit{Agence Nationale de la Recherche} (CoMoC project).\\
The authors thank Professor John P. Killingbeck for his help.

\appendix
\section{Mathematical properties of the eigenoperators}
\label{annexeA}
\begin{prop}
The mixed eigenstate $\rho_E^\alpha(x) = \|\phi_E^\alpha(x)\|^2_*$ satisfies the following equations
\begin{equation}
\mathcal L_x(\rho_E^\alpha(x)) = E(x) \rho_E^\alpha(x) - \rho_E^\alpha(x) E(x)^\dagger
\end{equation}
\begin{eqnarray}
& & \mathcal L_x\left(E(x) \rho_E^\alpha(x)-\rho_E^\alpha(x) E(x)^\dagger\right) \nonumber \\
& & \qquad = E(x) \mathcal L_x(\rho_E^\alpha(x)) - \mathcal L_x(\rho_E^\alpha(x)) E(x)^\dagger
\end{eqnarray}
\end{prop}

\begin{proof}
\begin{eqnarray}
\mathcal L(\rho_E) & = & \tr_{\E} \left([H,|\phi_E \rrangle \llangle \phi_E|] \right) \nonumber \\
& = & \tr_{\E} \left( |H\phi_E \rrangle \llangle \phi_E | - |\phi_E \rrangle \llangle H^\ddagger \phi_E | \right) \nonumber \\
%& = & \tr_{\E} \left( E |\phi_E \rrangle \llangle \phi_E| - |\phi_E \rrangle \llangle \phi_E | E^\dagger \right) \\
%& = & E \tr_{\E}(|\phi_E \rrangle \llangle \phi_E|) - \tr_{\E}(|\phi_E \rrangle \llangle \phi_E |) E^\dagger \\
& = & E \rho_E- \rho_E E^\dagger
\end{eqnarray}

\begin{eqnarray}
& & \mathcal L(E\rho_E - \rho_E E^\dagger) \nonumber \\
& & = \tr_{\E}\left([H,E|\phi_E \rrangle \llangle \phi_E|] - [H,|\phi_E \rrangle \llangle \phi_E|E^\dagger] \right) \nonumber \\
%& & = \tr_{\E}\left(HE|\phi_E \rrangle \llangle \phi_E| - E|\phi_E \rrangle \llangle \phi_E|H \right. \nonumber \\
%& & \quad \left. - H|\phi_E \rrangle \llangle \phi_E|E^\dagger + |\phi_E \rrangle \llangle \phi_E|E^\dagger H \right) \\
& & = \tr_{\E}\left(EH|\phi_E \rrangle \llangle \phi_E| - E|\phi_E \rrangle \llangle \phi_E|H \right. \nonumber \\
& & \quad \left. - H|\phi_E \rrangle \llangle \phi_E|E^\dagger + |\phi_E \rrangle \llangle \phi_E|HE^\dagger\right) \nonumber \\
%& & = E \tr_{\E}([H,|\phi_E \rrangle \llangle \phi_E|]) - \tr_{\E}([H,|\phi_E \rrangle \llangle \phi_E|]) E^\dagger \\
& & = E \mathcal L(\rho_E) - \mathcal L(\rho_E)E^\dagger
\end{eqnarray}
\end{proof}
This property can be used to find the eigenoperator and the mixed eigenstate if the Lindbladian is known without the knowledge of the precise description of the universe.

\begin{prop}
$E$ is almost selfadjoint in the following sense
\begin{equation}
\forall \alpha,\forall x \in U^\alpha \quad \tr_{\S}(\rho_E^\alpha(x)(E(x)-E(x)^\dagger)) = 0
\end{equation}
\end{prop}

\begin{proof}
\begin{eqnarray}
\tr_{\S}(\rho_E(E-E^\dagger)) %& = & \tr_{\S}(\rho_E E) - \tr_{\S}(\rho_E E^\dagger) \\
%& = & \tr_{\S}(\rho_E E) - \tr_{\S}( E^\dagger \rho_E) \\
& = & \tr_{\S}(\rho_E E - E^\dagger \rho_E) \nonumber \\
& = & \tr_{\S}(\mathcal L(\rho_E)) \nonumber \\
& = & \tr_{\S}\tr_{\E}([H,|\phi_E\rrangle\llangle \phi_E|]) \nonumber \\
& = & 0
\end{eqnarray}
\end{proof}

The following proposition permits us to find easily the eigenoperators and $*$-eigenvectors.
\begin{propo}
\label{spectre}
For all $E_0(x) \in \mathcal L(\S)$ such that $[E_0(x) \otimes 1_{\E},H(x)] = 0$, we have
\begin{equation}
H(x) \phi_{E_0,\lambda}^\alpha(x) = (E_0(x)+\lambda(x)1_{\S}) \phi^\alpha_{E_0,\lambda}
\end{equation}
where $\lambda \in \mathbb C$ and $\phi_{E_0,\lambda}^\alpha \in \S\otimes \E$ are solutions of the usual eigenvalue equation
\begin{equation}
(H(x)-E_0(x) \otimes 1_{\E}) \phi_{E_0,\lambda}^\alpha(x) = \lambda(x) \phi_{E_0,\lambda}^\alpha(x)
\end{equation}
\end{propo}
We can solve the eigenoperator equation by finding the subalgebra of $\mathcal L(\S)$ of the operators commuting with $H$, and for each element $E_0$, by diagonalizing $H-E_0 \otimes 1_{\E}$. 

\begin{prop}
$J^0_x$ (the orbit of the isotropy subgroup of $\rho_E^\alpha$ by $G_x$) is a normal subgroup of $G_x$.
\end{prop}

\begin{proof}
Let $j \in J^0_x$ and $g\in G_x$. By definition of $J^0_x$, $\exists g' \in G_x$ such that $jg'\rho_E(g')^\dagger j^\dagger = g'\rho_E(g')^\dagger$. We have then
\begin{equation}
 gjg^{-1} \left(gg'\rho_E(g')^\dagger g^\dagger \right) (g^\dagger)^{-1} j^\dagger g^\dagger = gg'\rho_E(g')^\dagger g^\dagger 
\end{equation}
Then $gjg^{-1} \in gg'G_{x,\rho_E} (g')^\dagger g^\dagger \subset J_x^0$. We have then $gJ_x^0g^{-1} \subset J_x^0$, $J_x^0$ is normal.
\end{proof}

\begin{prop}
$\mathfrak j^1_x$ is a solvable subalgebra of $\mathfrak g_x$.
\end{prop}

\begin{proof}
$\forall O_1,O_2 \in \mathfrak j_x^1$
\begin{eqnarray}
\widehat{[O_1,O_2]} & = & \left(\begin{array}{cc} 0 & *_{11} \\ 0 & *_{12} \end{array} \right)\left(\begin{array}{cc} 0 & *_{21} \\ 0 & *_{22} \end{array} \right) - \left(\begin{array}{cc} 0 & *_{21} \\ 0 & *_{22} \end{array} \right)\left(\begin{array}{cc} 0 & *_{11} \\ 0 & *_{12} \end{array} \right) \nonumber \\
& = & \left(\begin{array}{cc} 0 & *_{11}*_{22} - *_{21}*_{12} \\ 0 & 0 \end{array} \right)
\end{eqnarray}
then
\begin{equation}
(\mathfrak j_x^1)' = \left\{O \in \mathfrak g_x, \hat O = \left(\begin{array}{cc} 0 & * \\ 0 & 0 \end{array} \right) \right\}
\end{equation}
$\forall O_1,O_2\in (\mathfrak j_x^1)'$
\begin{eqnarray}
\widehat{[O_1,O_2]} & = & \left(\begin{array}{cc} 0 & *_{1} \\ 0 & 0 \end{array} \right) \left(\begin{array}{cc} 0 & *_2 \\ 0 & 0 \end{array} \right) - \left(\begin{array}{cc} 0 & *_2 \\ 0 & 0 \end{array} \right)  \left(\begin{array}{cc} 0 & *_{1} \\ 0 & 0 \end{array} \right) \nonumber \\
& = & \left(\begin{array}{cc} 0 & 0 \\ 0 & 0 \end{array} \right)
\end{eqnarray}
then $(\mathfrak j_x^1)'' = \{0\}$.
\end{proof}

\begin{prop}
$E(x) \in \mathfrak g_x$
\end{prop}

\begin{proof}
\begin{equation}
HE\phi_E = EH\phi_E = EE\phi_E \Rightarrow E\phi_E \in \ker(H-E) \Rightarrow E \in \mathfrak g_x
\end{equation}
\end{proof}

\begin{prop}
\label{inner}
We assume that $E$ is non-degenerate. Let $U^\alpha \ni x \mapsto \psi^\alpha(x) \in \ker(H(x)-E(x)\otimes 1_{\E})$ and $U^\alpha \ni x \mapsto \phi^\alpha(x) \in \ker(H(x)-E(x) \otimes 1_{\E})$ be two continuous $*$-eigenvectors, then $\langle \phi^\alpha|\psi^\alpha \rangle_* \in \Gamma(U^\alpha,P^\alpha_{\mathfrak g})$ ($\Gamma(U^\alpha,P^\alpha_{\mathfrak g})$ denotes the set of the sections of $P^\alpha_{\mathfrak g}$).
\end{prop}

\begin{proof}
Since the action of $G_x \times K_x$ is transitive, we have $\mathcal{GL}(\ker(H(x)-E(x)\otimes 1_{\E})) = G_x \times K_x$. This induces that $\mathcal{L}(\ker(H(x)-E(x)\otimes 1_{\E})) = \mathrm{Env}\left(\mathfrak g_x \oplus \mathfrak k_x\right)$ ($\mathrm{Env}(\mathfrak a)$ denotes the universal enveloping algebra of $\mathfrak a$). $\mathfrak g_x$ and $\mathfrak k_x$ are operators Lie algebras which are stable by operator composition, we then have $\mathrm{Env}\left(\mathfrak g_x \oplus \mathfrak k_x \right)= \mathbb C \oplus \mathfrak g_x \oplus \mathfrak k_x \oplus \mathfrak g_x \otimes \mathfrak k_x$.\\
$\exists X_i \in \mathfrak g_x$ and $\exists Y_i \in \mathfrak k_x$ such that
\begin{eqnarray}
|\psi \rrangle \llangle \phi | & = & \sum_i X_i \otimes Y_i \nonumber \\
 \Rightarrow \langle \phi|\psi \rangle_* & = & \sum_i \tr_{\E}(Y_i)X_i \in \mathfrak g_x
\end{eqnarray}
\end{proof}

\section{Mathematical properties of the $C^*$-geometric phase generator}
\label{annexeB}
\begin{prop}
The $C^*$-geometric phase generator is almost antiselfadjoint in the following sense
\begin{equation}
\forall \alpha, \forall x \in U^\alpha \quad \tr_{\S}\left(\rho_E^\alpha(x)(\mathcal A(x)+\mathcal A(x)^\dagger) \right) = 0
\end{equation}
\end{prop}

\begin{proof}
\begin{eqnarray}
\tr_{\S}(\rho_E \mathcal A) & = & \tr_{\S}(\mathcal A \rho_E) \nonumber\\
& = & \tr_{\S}(\langle \phi_E|d\phi_E\rangle_*) \nonumber \\
%& = & \tr_{\S} \tr_{\E} (|d\phi_E \rrangle \llangle \phi_E|) \\
& = & \llangle \phi_E|d\phi_E \rrangle \nonumber \\
& = & - \llangle d\phi_E | \phi_E \rrangle \nonumber \\
& = & - \tr_{\S}(\langle \phi_E|d\phi_E \rangle_*^\dagger) \nonumber \\
& = & - \tr_{\S}(\rho_E \mathcal A^\dagger)
\end{eqnarray}
\end{proof}

\begin{prop}
Assuming that $\rho_E^\alpha$ is invertible, then the $C^*$-geometric phase generator is 
\begin{equation}
\breve {\mathcal A}^\alpha = \frac{1}{2} d\rho_E^\alpha (\rho_E^\alpha)^{-1}
\end{equation}
modulo a $\mathfrak j_x$-gauge transformation. 
\end{prop}

\begin{proof}
\begin{eqnarray}
d\rho_E & = & \frac{1}{2}(d\rho_E \rho_E^{-1} \rho_E + \rho_E \rho_E^{-1} d\rho_E) \nonumber\\
& = & \frac{1}{2}(d\rho_E \rho_E^{-1} \rho_E + \rho_E (d\rho_E \rho_E^{-1})^\dagger)
\end{eqnarray}
$\breve {\mathcal A}$ is then a solution of the equation (\ref{geomphasematdens}) like $\mathcal A = \langle \phi_E|d\phi_E\rangle_* \|\phi_E\|^{-2}_*$. We have then
\begin{equation}
(\mathcal A - \breve{\mathcal A}) \rho_E + \rho_E (\mathcal A - \breve{\mathcal A})^\dagger = 0
\end{equation}
and then $(\mathcal A - \breve{\mathcal A}) \in \Omega^1(U^\alpha, P^\alpha_{\mathfrak j_x})$.
\end{proof}
If $\rho_E^\alpha$ is not invertible, equation (\ref{geomphase}) can have several solutions, but we have the following result
\begin{prop}
Assuming that $\ker \rho^{\alpha}_E \not= \{0\}$, if equation (\ref{geomphase}) has some solutions then all these solutions are equals modulo $\mathfrak j_x$-gauge transformations.
\end{prop}

\begin{proof}
The matrix representation of $\|\phi_E\|^2_*$ in its diagonalization basis is
\begin{equation}
 \widehat{\|\phi_E\|^2_*} = \left(\begin{array}{cc} \varphi_{p\times p} & 0_{p\times (n-p)} \\ 0_{(n-p)\times p} & 0_{(n-p) \times (n-p)} \end{array}\right)
\end{equation}
where $\varphi_{p\times p}$ is a diagonal matrix of order $p$. Thus, we have
\begin{equation}
\hat {\mathcal A} \left(\begin{array}{cc} \varphi_{p\times p} & 0_{p\times (n-p)} \\ 0_{(n-p)\times p} & 0_{(n-p) \times (n-p)} \end{array}\right) = \widehat{\langle \phi_E|d\phi_E\rangle_*}
\end{equation}
Equation (\ref{geomphase}) has some solutions if and only if in the diagonalization basis of $\rho_E$ the operator $\langle \phi_E|d\phi_E\rangle_*$ is represented by a matrix of the form $\left(\begin{array}{cc} X_{p \times p} & 0_{p \times (n-p)} \\ Y_{(n-p)\times p} & 0_{(n-p) \times (n-p)} \end{array} \right)$. The solutions are then
\begin{equation}
\hat {\mathcal A} = \left(\begin{array}{cc} X_{p \times p} \varphi_{p \times p}^{-1} & C_{p \times (n-p)} \\ Y_{(n-p)\times p} \varphi_{p \times p}^{-1} & C_{(n-p) \times (n-p)} \end{array} \right)
\end{equation}
where $\hat C = \left(\begin{array}{cc}  0_{p \times p} & C_{p \times (n-p)} \\ 0_{(n-p)\times p} & C_{(n-p) \times (n-p)} \end{array} \right)$ is an arbitrary matrix. By construction the associated operator $C$ belongs to $\mathfrak j_x$.
\end{proof}
We note an important difference between the two previous properties; in the first one the gauge transformation belongs to $\mathfrak j^0_x$ whereas in the second one it belongs to $\mathfrak j^1_x$.

\begin{prop}
Assuming that $\rho^\alpha_E$ is invertible, and let $\breve {\mathcal A}^\alpha = \frac{1}{2} d\rho^\alpha_E (\rho^\alpha_E)^{-1}$, then under a gauge transformation $\tilde \rho^\alpha_E = g \rho^\alpha_E g^\dagger$ with $g \in \Gamma(U^\alpha,P^\alpha_G)$ we have
\begin{equation}
\tilde{\breve{\mathcal A}^\alpha} = g \breve{\mathcal A}^\alpha g^{-1}+ dgg^{-1} + g\breve \eta g^{-1}
\end{equation}
with
\begin{equation}
\breve \eta = \frac{1}{2} \left(\rho_E^\alpha dg^\dagger (g^\dagger)^{-1} (\rho_E^\alpha)^{-1} - g^{-1}dg \right) \in \Omega^1(U^\alpha,P^\alpha_{\mathfrak j})
\end{equation}
\end{prop}

\begin{proof}
\begin{equation}
d\tilde \rho_E = dg \rho_E g^\dagger + g d\rho_E g^\dagger + g\rho_E dg^\dagger 
\end{equation}
\begin{eqnarray}
\tilde {\breve {\mathcal A}} & = & \frac{1}{2} d\tilde \rho_E \tilde \rho_E^{-1} \nonumber \\
& = & \frac{1}{2}(dgg^{-1} + gd\rho_E \rho_E^{-1}g^{-1} + g \rho_E dg^{\dagger}(g^\dagger)^{-1} \rho_E^{-1} g^{-1}) \nonumber \\
& = & dgg^{-1} + g\breve{\mathcal A} g^{-1} + \frac{1}{2}g(\rho_E dg^{\dagger}(g^\dagger)^{-1} \rho_E^{-1} - g^{-1}dg)g^{-1}
\end{eqnarray}
\begin{eqnarray}
& & \breve \eta \rho_E + \rho_E \breve \eta^\dagger \nonumber \\
& & \quad =  \frac{1}{2}(\rho_E dg^\dagger (g^\dagger)^{-1}-g^{-1}dg \rho_E +g^{-1}dg \rho_E - \rho_E dg^\dagger (g^\dagger)^{-1}) \nonumber \\
& & \quad = 0
\end{eqnarray}
then $\breve \eta \in \Omega^1(U^\alpha,P^\alpha_{\mathfrak j_x})$.
\end{proof}

\section{The categorical bundle and its connective structure}
\label{annexeC}
\subsection{Curving, fake curvature and true curvature}
Let $B^\alpha = d\mathcal A^\alpha - \mathcal A^\alpha \wedge \mathcal A^\alpha \in \Omega^2(U^\alpha,\mathcal L(\S))$ be the curving and $F^\alpha = dR^\alpha - [A^\alpha,R^\alpha] - R^\alpha \wedge R^\alpha \in \Omega^2(U^\alpha,\mathcal L(\S))$ be the fake curvature.
\begin{equation}
B^\alpha = dA^\alpha - A^\alpha \wedge A^\alpha + F^\alpha
\end{equation}

\begin{prop}
$B^\alpha \in \Omega^2(U^\alpha,P^\alpha_{\mathfrak j})$ and $F^\alpha \in \Omega^2(U^\alpha,P^\alpha_{\mathfrak g})$.
\end{prop}

\begin{proof}
By differentiating equation (\ref{geomphase}) we find
\begin{equation}
d\mathcal A^\alpha \|\phi^\alpha_E\|^2_* - \mathcal A^\alpha \wedge d\|\phi^\alpha_E\|^2_* = -\tr_{\E}(|d\phi^\alpha_E\rrangle\wedge \llangle d\phi^\alpha_E|) 
\end{equation}
\begin{eqnarray}
& & d\mathcal A^\alpha \|\phi^\alpha_E\|^2_* - \mathcal A^\alpha \wedge \mathcal A^\alpha \|\phi^\alpha_E\|^2_* - \mathcal A^\alpha \wedge \|\phi^\alpha_E\|^2_* (\mathcal A^\alpha)^\dagger \nonumber \\
& & \qquad = -\tr_{\E}(|d\phi^\alpha_E\rrangle\wedge \llangle d\phi^\alpha_E|)
\end{eqnarray}
\begin{equation}
B^\alpha \|\phi^\alpha_E\|^2_* =  \mathcal A^\alpha \wedge \|\phi^\alpha_E\|_*^2 (\mathcal A^\alpha)^\dagger -\tr_{\E}(|d\phi^\alpha_E\rrangle\wedge \llangle d\phi^\alpha_E|)
\end{equation}
Taking the transconjugation of this last expression we find
\begin{eqnarray}
\|\phi^\alpha_E\|^2_* (B^\alpha)^\dagger & = & (\mathcal A^\alpha_\mu \|\phi^\alpha_E\|^2_* (\mathcal A^\alpha_\nu)^\dagger)^\dagger dx^\mu \wedge dx^\nu \nonumber \\
& & \quad - \tr_{\E}(|\partial_\mu \phi^\alpha_E \rrangle \llangle \partial_\nu \phi^\alpha_E|)^\dagger dx^\mu \wedge dx^\nu
\end{eqnarray}
\begin{eqnarray}
\|\phi^\alpha_E\|^2_* (B^\alpha)^\dagger & = & \mathcal A^\alpha_\nu \|\phi^\alpha_E\|^2_* (\mathcal A^\alpha_\mu)^\dagger dx^\mu \wedge dx^\nu \nonumber \\
& & \quad - \tr_{\E}(|\partial_\nu \phi^\alpha_E \rrangle \llangle \partial_\mu \phi^\alpha_E|) dx^\mu \wedge dx^\nu
\end{eqnarray}
\begin{equation}
\|\phi^\alpha_E\|^2_* (B^\alpha)^\dagger = -\mathcal A^\alpha \wedge \|\phi^\alpha_E \|^2_*\mathcal A^\alpha + \tr_{\E}(|d\phi^\alpha_E\rrangle\wedge \llangle d\phi^\alpha_E|)
\end{equation}
Then, we have
\begin{equation}
B^\alpha \|\phi^\alpha_E\|^2_* + \|\phi^\alpha_E\|^2_* (B^\alpha)^\dagger = 0
\end{equation}
\begin{equation}
\Rightarrow B^\alpha \in \Omega^2(U^\alpha,P^\alpha_{\mathfrak j})
\end{equation}
By definition $F^\alpha = B^\alpha - dA^\alpha + A^\alpha \wedge A^\alpha$, $A^\alpha \in \Omega^1(U^\alpha,P^\alpha_{\mathfrak g})$ and $B^\alpha \in \Omega^2(U^\alpha,P^\alpha_{\mathfrak j}) \subset \Omega^2(U^\alpha,P^\alpha_{\mathfrak g})$, then $F^\alpha \in \Omega^2(U^\alpha,P^\alpha_{\mathfrak g})$.
\end{proof}
If $\rho_E^\alpha$ is invertible, the curving associated with $\breve A^\alpha$ is $\breve B^\alpha = \frac{1}{4} d\rho_E^\alpha (\rho_E^\alpha)^{-1} \wedge d\rho_E^\alpha (\rho_E^\alpha)^{-1}$.\\
Finally we set $H^\alpha \in \Omega^3(U^\alpha,P^\alpha_{\mathfrak j})$ the true curvature defined by
\begin{eqnarray}
H^\alpha & = & dB^\alpha - [A^\alpha,B^\alpha] \nonumber \\
& = & dF^\alpha - [A^\alpha,F^\alpha]
\end{eqnarray}

Let us now consider the gauge transformations at the intersection of several charts.
\begin{prop}
$\forall x \in U^\alpha \cap U^\beta$, the curving satisfies
\begin{eqnarray}
B^\beta & = & (g^{\alpha \beta})^{-1} B^\alpha g^{\alpha \beta} \nonumber \\
& & \quad + (g^{\alpha \beta})^{-1} \left(d\eta^{\alpha \beta} - \eta^{\alpha \beta} \wedge \eta^{\alpha \beta} - [A^\alpha,\eta^{\alpha \beta}] + \chi^{\alpha \beta} \right)g^{\alpha \beta}
\end{eqnarray}
with $\chi^{\alpha \beta} = [R^\alpha,\eta^{\alpha \beta}] \in \Omega^2(U^\alpha \cap U^\beta) \otimes_M \bigsqcup_x \mathfrak j_x$ called the curving-transformation.
\end{prop}

\begin{proof}
From the gauge transformation formula of the $C^*$-geometric phase generator we have
\begin{equation}
\eta^{\alpha \beta} = g^{\alpha \beta} \mathcal A^\beta (g^{\alpha \beta})^{-1} - \mathcal A^\alpha + dg^{\alpha \beta}(g^{\alpha \beta})^{-1} 
\end{equation}
then
\begin{eqnarray}
d\eta^{\alpha \beta} - \eta^{\alpha \beta} \wedge \eta^{\alpha \beta} %& = & dg^{\alpha \beta}(g^{\alpha \beta})^{-1} \wedge g^{\alpha \beta} \mathcal A^\beta (g^{\alpha \beta})^{-1} \nonumber \\
%& & + g^{\alpha \beta} d\mathcal A^\beta (g^{\alpha \beta})^{-1} \nonumber\\
%& & + g^{\alpha \beta} \mathcal A^\beta (g^{\alpha \beta})^{-1} \wedge dg^{\alpha \beta}(g^{\alpha \beta})^{-1} \nonumber\\
%& & - d\mathcal A^\alpha \nonumber\\
%& & + dg^{\alpha \beta}(g^{\alpha \beta})^{-1} \wedge dg^{\alpha \beta}(g^{\alpha \beta})^{-1} \nonumber\\
%& & - g^{\alpha \beta} \mathcal A^\beta \wedge \mathcal A^\beta (g^{\alpha \beta})^{-1} \nonumber\\
%& & + g^{\alpha \beta} \mathcal A^\beta (g^{\alpha \beta})^{-1} \wedge \mathcal A^\alpha \nonumber\\
%& & - g^{\alpha \beta} \mathcal A^\beta (g^{\alpha \beta})^{-1} \wedge dg^{\alpha \beta}(g^{\alpha \beta})^{-1} \nonumber\\
%& & + \mathcal A^\alpha \wedge g^{\alpha \beta} \mathcal A^\beta (g^{\alpha \beta})^{-1} \nonumber\\
%& & - \mathcal A^\alpha \wedge \mathcal A^\alpha \nonumber\\
%& & + \mathcal A^\alpha \wedge dg^{\alpha \beta}(g^{\alpha \beta})^{-1} \nonumber\\
%& & - dg^{\alpha \beta}(g^{\alpha \beta})^{-1} \wedge g^{\alpha \beta}\mathcal A^\beta (g^{\alpha \beta})^{-1}\nonumber\\
%& & + dg^{\alpha \beta}(g^{\alpha \beta})^{-1} \wedge \mathcal A^\alpha\nonumber \\
%& & - dg^{\alpha \beta}(g^{\alpha \beta})^{-1} \wedge dg^{\alpha \beta}(g^{\alpha \beta})^{-1} \\
& = & g^{\alpha \beta} B^\beta (g^{\alpha \beta})^{-1} - B^\alpha - 2 \mathcal A^\alpha \wedge \mathcal A^\alpha \nonumber\\
& & + \mathcal A^\alpha \wedge g^{\alpha \beta} \mathcal A^\beta (g^{\alpha \beta})^{-1} + \mathcal A^\alpha \wedge dg^{\alpha \beta}(g^{\alpha \beta})^{-1} \nonumber\\
& & + g^{\alpha \beta}(g^{\alpha \beta})^{-1} \wedge \mathcal A^\alpha + g^{\alpha \beta} \mathcal A^\beta (g^{\alpha \beta})^{-1} \wedge \mathcal A^\alpha \nonumber \\
& = & g^{\alpha \beta} B^\beta (g^{\alpha \beta})^{-1} - B^\alpha \nonumber \\
& & + \mathcal A^\alpha \wedge \eta^{\alpha \beta} + \eta^{\alpha \beta} \wedge \mathcal A^\alpha
\end{eqnarray}
\end{proof}

\begin{prop}
$\forall x \in U^\alpha \cap U^\beta$, the fake curvature satisfies
\begin{equation}
F^\beta = (g^{\alpha \beta})^{-1} F^\alpha g^{\alpha \beta} - (g^{\alpha \beta})^{-1} \chi^{\alpha \beta} g^{\alpha \beta}
\end{equation}
\end{prop}

\begin{proof}
\begin{eqnarray}
F^\beta & = & d((g^{\alpha \beta})^{-1} R^\alpha g^{\alpha \beta}) \nonumber \\
& & - \left[(g^{\alpha \beta})^{-1}A^\alpha g^{\alpha \beta}+d(g^{\alpha \beta})^{-1}g^{\alpha \beta}+(g^{\alpha \beta})^{-1}\eta^{\alpha \beta} g^{\alpha \beta}, \right. \nonumber \\
& & \qquad \left. (g^{\alpha \beta})^{-1}R^\alpha g^{\alpha \beta} \right] -(g^{\alpha \beta})^{-1} R^\alpha \wedge R^\alpha g^{\alpha \beta} \nonumber \\
%& = & d(g^{\alpha \beta})^{-1} g^{\alpha \beta} \wedge (g^{\alpha \beta})^{-1} R^\alpha g^{\alpha \beta} + (g^{\alpha \beta})^{-1} dR^\alpha g^{\alpha \beta} \nonumber \\
%& & + (g^{\alpha \beta})^{-1} R^\alpha g^{\alpha \beta} \wedge d(g^{\alpha \beta})^{-1} g^{\alpha \beta} \nonumber \\
%& & -(g^{\alpha \beta})^{-1} [A^\alpha,R^\alpha] g^{\alpha \beta} - [d(g^{\alpha \beta})^{-1}g^{\alpha \beta},(g^{\alpha \beta})^{-1}R^\alpha g^{\alpha \beta}] \nonumber \\
%& & - (g^{\alpha \beta})^{-1} [\eta^{\alpha \beta},R^\alpha] g^{\alpha \beta} - (g^{\alpha \beta})^{-1} R^\alpha \wedge R^\alpha g^{\alpha \beta} \\
& = & (g^{\alpha \beta})^{-1} F^\alpha g^{\alpha \beta} - (g^{\alpha \beta})^{-1} \chi^{\alpha \beta} g^{\alpha \beta}
\end{eqnarray}
\end{proof}

\begin{prop}
$\forall x \in U^\alpha \cap U^\beta \cap U^\gamma$, the curving-transformation satisfies
\begin{equation}
 \chi^{\alpha \beta} + g^{\alpha \beta} \chi^{\beta \gamma} (g^{\alpha \beta})^{-1} - h^{\alpha \beta \gamma} \chi^{\alpha \gamma} (h^{\alpha \beta \gamma})^{-1} = - h^{\alpha \beta \gamma} [F^\alpha,(h^{\alpha \beta \gamma})^{-1}]
\end{equation}
\end{prop}

\begin{proof}
\begin{eqnarray}
\chi^{\alpha \beta} & = & F^\alpha - g^{\alpha \beta} F^\beta (g^{\alpha \beta})^{-1} \nonumber \\
& = & F^\alpha - g^{\alpha \beta} \left(\chi^{\beta \gamma}+g^{\beta \gamma}F^\gamma(g^{\beta \gamma})^{-1} \right) (g^{\alpha \beta})^{-1} \nonumber \\
%& = & F^\alpha - g^{\alpha \beta} \left(\chi^{\beta \gamma}+g^{\beta \gamma}\left(\chi^{\gamma \alpha} \right.\right. \nonumber \\
%& & \quad \left.\left. +g^{\gamma \alpha}F^\alpha (g^{\gamma \alpha})^{-1}\right)(g^{\beta \gamma})^{-1}\right)(g^{\alpha \beta})^{-1} \\
& = & F^\alpha - g^{\alpha \beta} \chi^{\beta \gamma} (g^{\alpha \beta})^{-1} - g^{\alpha \beta} g^{\beta \gamma} \chi^{\gamma \alpha} (g^{\beta \gamma})^{-1} (g^{\alpha \beta})^{-1} \nonumber \\
& & \quad - h^{\alpha \beta \gamma} F^\alpha (h^{\alpha \beta \gamma})^{-1}
\end{eqnarray}
but
\begin{eqnarray}
\chi^{\gamma \alpha} & = & F^\gamma - g^{\gamma \alpha} F^\alpha (g^{\gamma \alpha})^{-1} \nonumber \\
& = & g^{\gamma \alpha}( g^{\alpha \gamma} F^\gamma (g^{\alpha \gamma})^{-1} - F^\alpha) (g^{\gamma \alpha})^{-1} \nonumber\\
& = & - g^{\gamma \alpha} \chi^{\alpha \gamma} (g^{\gamma \alpha})^{-1}
\end{eqnarray}
then
\begin{equation}
 \chi^{\alpha \beta} = F^\alpha - g^{\alpha \beta} \chi^{\beta \gamma} (g^{\alpha \beta})^{-1} + h^{\alpha \beta \gamma} \chi^{\alpha \gamma} (h^{\alpha \beta \gamma})^{-1} - h^{\alpha \beta \gamma} F^\alpha (h^{\alpha \beta \gamma})^{-1}
\end{equation}
\end{proof}

\subsection{The fibred structure}
\subsubsection{Construction of $\mathcal P$ as a non-abelian bundle gerbes:}
Let $\xi^{\alpha \beta} : P^\beta_{G|U^\alpha \cap U^\beta} \to P^\alpha_{G|U^\alpha \cap U^\beta}$ be the bundle isomorphism defined by
\begin{equation}
\forall p \in P^\beta_{G|U^\alpha \cap U^\beta}, \quad \xi^{\alpha \beta}(p) = g^{\alpha \beta}(\pi_{P^\beta_G}(p))p
\end{equation}
where $\pi_{P^\beta_G} : P^\beta_G \to U^\beta$ is the canonical projection of the bundle $P^\beta_G$, the product between $g^{\alpha \beta}(\pi_{P^\beta_G}(p))$ and $p$ is the group law of $G_{\pi_{P^\beta_G}(p)} \subset \mathcal{GL}(\S)$.\\
Let $\mathring g^{\alpha \beta} \in \mathcal C^\infty(U^\alpha \cap U^\beta,G)$ be the trivialization of $g^{\alpha \beta}$, i.e.
\begin{equation}
\mathring g^{\alpha \beta}(x) = (\zeta^\alpha_x)^{-1} \circ \xi^{\alpha \beta}_x \circ \zeta^\beta_x(1_G) = (\zeta^\alpha_x)^{-1} \left( g^{\alpha \beta}(x) \zeta^\beta_x(1_G)  \right)
\end{equation}
in other words, we have the following commutative diagram:
$$ \begin{CD}
G @>{\mathring g^{\alpha \beta}(x) \times}>> G \\
@V{\zeta^\beta_x}VV @VV{\zeta^\alpha_x}V \\
\pi^{-1}_{P^\beta_G}(x) @>>{\xi^{\alpha \beta}_x}> \pi^{-1}_{P^\alpha_G}(x) 
\end{CD} $$
We define $\mathring h^{\alpha \beta \gamma} \in \mathcal C^\infty(U^\alpha \cap U^\beta \cap U^\gamma,J)$ by
\begin{equation}
\forall p \in \pi^{-1}_{P^\alpha_G}(x), \quad \xi^{\alpha \beta}_x \circ \xi^{\beta \gamma}_x \circ \xi^{\gamma \alpha}_x (p) = \mathring h^{\alpha \beta \gamma}(x) p
\end{equation}
in other words
\begin{equation}
\mathring h^{\alpha \beta \gamma}(x) = (\zeta^\alpha_x)^{-1}(h^{\alpha \beta \gamma}(x) \zeta^{\alpha}_x(1_G))
\end{equation}
Let $\zeta^\alpha_{x,Lie} : \mathfrak g \to \mathfrak g_x$ be the map induced in the Lie algebras by $\zeta^\alpha_x$ (the local trivialization of $P^\alpha_G$). We extend $\zeta^\alpha$ on the $\mathfrak g$-valued differential forms of $U^\alpha$ by $ \zeta^\alpha : \Omega^*(U^\alpha, \mathfrak g) \to \Omega^*(U^\alpha,P^\alpha_{\mathfrak g})$
\begin{eqnarray}
& & \zeta^\alpha\left(a^i_{\mu_1...\mu_p}(x)X_i dx^{\mu_1} \wedge ... \wedge dx^{\mu_p}\right) \nonumber \\
& & \qquad =  a^i_{\mu_1...\mu_p}(x) \zeta^\alpha_{x,Lie}(X_i) dx^{\mu_1} \wedge ... \wedge dx^{\mu_p}
\end{eqnarray}
where $\{X_i\}_i$ is the set of the generators of $\mathfrak g$.

\begin{prop}
$\zeta^\alpha$ is a chain map and an algebra homomorphism : $\forall a,b \in \Omega^*(U^\alpha,\mathfrak g)$
\begin{equation}
d\zeta^\alpha(a) = \zeta^\alpha(da)
\end{equation}
\begin{equation}
\zeta^\alpha(a\wedge b) = \zeta^\alpha(a) \wedge \zeta^\alpha(b)
\end{equation}
the derivative on $\Omega^*(U^\alpha,P^\alpha_{\mathfrak g}) = \Omega^*(U^\alpha) \otimes_M \Gamma(U^\alpha,P^\alpha_{\mathfrak g})$ being defined as $d \otimes 1_{P^\alpha_{\mathfrak g}}$.
\end{prop}

\begin{proof}
This is a direct consequence of the definition of the extension of $\zeta^\alpha$.
\end{proof}
We are then able to define the trivializations of the different forms : $\mathring A^\alpha = (\zeta^\alpha)^{-1}(A^\alpha) \in \Omega^1(U^\alpha,\mathfrak g)$, $\mathring B^{\alpha} = (\zeta^\alpha)^{-1}(B^\alpha) \in \Omega^2(U^\alpha,\mathfrak j)$, $\mathring F^\alpha = (\zeta^\alpha)^{-1}(F^\alpha) \in \Omega^2(U^\alpha,\mathfrak g)$, $\mathring \eta^{\alpha \beta} = (\zeta^{\alpha})^{-1}(\eta^{\alpha \beta}) \in \Omega^1(U^\alpha \cap U^\beta,\mathfrak j)$, $\mathring \chi^{\alpha \beta} = (\zeta^\alpha)^{-1}(\chi^{\alpha \beta}) \in \Omega^2(U^\alpha \cap U^\beta,\mathfrak j)$, and $\mathring H^\alpha = (\zeta^\alpha)^{-1}(H^\alpha) \in \Omega^3(U^\alpha,\mathfrak j)$.\\

The geometric structure $\mathcal P$ is defined by the 1-transition $\mathring g^{\alpha \beta}$ and the 2-transition $\mathring h^{\alpha \beta \gamma}$. It is endowed with a 2-connection (see ref. \cite{baez1,baez2,baez2}) defined by the gauge potential $\mathring A^\alpha$, the curving $\mathring B^\alpha$ and the potential-transformation $\mathring \eta^{\alpha \beta}$. Its geometry is characterized by the fake curvature $\mathring F^\alpha$, the curving-transformation $\mathring \chi^{\alpha \beta}$ and the true curvature $\mathring H^\alpha$. All these data define $\mathcal P$ as a non-abelian bundle gerbes \cite{breen,aschieri,kalkkinen,baez1,baez2,baez3}) with the structure Lie crossed module $(G,J,t,\Ad)$ where $t:J \to G$ is the canonical injection ($J$ is a subgroup of $G$) and $\Ad : G \to \mathrm{Aut}(J)$ is the adjoint representation of $G$ on itself restricted to $J$ (in the homomorphisms domain).

\subsubsection{Construction of $\mathcal P$ as a categorical bundle:}
$\mathcal P$ can be viewed as a principal categorical bundle (a 2-bundle, see ref. \cite{baez1,baez2,baez3}) on the left with the structure groupo\"id $\mathcal G$ having $\Obj(\mathcal G) = G$ as set of objects and $\Morph(\mathcal G) = J \rtimes G$ as set of arrows, the semi-direct product (the arrows horizontal composition) being defined by $(h,g)(h',g') = (h\Ad(g)h',gg')$. The source map of $\mathcal G$ is defined by $s(h,g) = g$ and the target map is defined by $t(h,g) = t(h)g$ (with $t(h)$ the canonical injection of $h$ in $G$). The morphisms composition (the arrows vertical composition) is defined by $(h,g) \circ (h',t(h)g) = (hh',g)$.\\
The total category $\mathcal P$ is then defined by $\Obj(\mathcal P) = \bigsqcup_\alpha P^\alpha_G$ and $\Morph(\mathcal P) = \bigsqcup_\alpha P^\alpha_J \times_M P^\alpha_G$ with $\forall (q,p) \in P^\alpha_J \times_M P^\alpha_G$ the source map $s(q,p) = p$ and the target map $t(q,p) = (\zeta^\alpha_{\pi_{P^\alpha_J}(q)})^{-1}(q)p$. The arrows composition is $(q,p) \circ (q', \zeta^\alpha_{\pi_{P^\alpha_J}(q)})^{-1}(q)p) = \left(qq',p\right)$. Let $\mathcal M$ be the representation of $M$ as a category, i.e. $\Obj(\mathcal M) = M$ and $\Morph(\mathcal M) = \{id_x : x \to x \}_{x \in M}$. The projection functor $\pi_{\mathcal P} : \mathcal P \to \mathcal M$ is defined by $\pi_{\mathcal P}^{\Obj} = \bigsqcup_\alpha \pi_{P^\alpha_G}$ and $\pi_{\mathcal P}^{\Morph}(q,p) = id_{\pi_{P^\alpha_G}(p)}$. The local trivialization functor $\vartheta^\alpha : \mathcal U^\alpha \times \mathcal G \to \mathcal P_{|U^\alpha}$ is defined by $\vartheta^{\alpha,\Obj} = \zeta^\alpha$ and $\vartheta^{\alpha,\Morph} = \zeta^\alpha_{|U^\alpha \times J} \times \zeta^\alpha$ ($\mathcal U^\alpha$ is the representation of $U^\alpha$ as a category). Finally the left action functor $L : \mathcal G \to \mathrm{Funct}(\mathcal P,\mathcal P)$ of $\mathcal G$ on $\mathcal P$ ($\mathrm{Funct}(\mathcal P,\mathcal P)$ is the category of the endofunctors of $\mathcal P$) is defined by: $\forall g \in \Obj(\mathcal G) = G$
\begin{equation}
L(g) : \begin{array}{rcl} \Obj(\mathcal P) & \to & \Obj(\mathcal P) \\ P^\alpha_G \ni p & \mapsto & L_{G}^\alpha(g) p \end{array}
\end{equation}
and $\forall (h,g) \in \Morph(\mathcal G) = J \rtimes G$
\begin{equation}
L(h,g) : \begin{array}{rcl} \Morph(\mathcal P) & \to & \Morph(\mathcal P) \\ P^\alpha_J \times_M P^\alpha_G \ni (q,p) & \mapsto & (L_{J}^\alpha(h) \Ad(L_G^\alpha(g))q, L_G^\alpha(g)p) \end{array}
\end{equation}
where $L^\alpha_J$ and $L^\alpha_G$ are the canonical left actions of $J$ and $G$ on $P^\alpha_J$ and $P^\alpha_G$.\\

Finally, we can note that because $G/J$ is a group since $J$ is normal, $\mathcal P$ is also a non-abelian twisted bundle \cite{mackaay} associated with the extension of groups:
$$ 1 \to J \to G \to G/J \to 1 $$

\section{About the adiabaticity of mixed states}
\label{annexeD}
Even if this is not strictly necessary, we consider in this section that the eigenoperator $E$ is selfadjoint, in order to avoid unnecessary complications involving biorthogonal basis.
\subsection{The adiabatic approximation}
The following theorem shows that there exists an adiabatic approximation for the $*$-eigenvectors.
\begin{theo}[Adiabatic approximation]
\label{thadiab}
Let $M \ni x \mapsto E(x) = E_0(x) + \lambda(x) 1_{\S} \in \mathcal L(\S)$ be a continuous selfadjoint eigenoperator decomposed following proposition \ref{spectre}, and $U^\alpha \ni x \mapsto \phi^\alpha_{E_0,\lambda}(x) \in \S \otimes \E$ an associated $*$-eigenvector. Let $t \mapsto x(t) \in U^\alpha$ be an evolution. We assume the following adiabatic condition
\begin{eqnarray}
& & \forall \mu \in \Sp(H-E_0 \otimes 1_{\E}) \setminus \{\lambda\}, \nonumber \\
& & \forall \phi_{E_0,\mu} \in \ker(H-E_0 \otimes 1_{\E} - \mu 1_{\S \otimes \E}) \nonumber \\
& & \quad  \frac{\llangle \phi_{E_0,\mu}(x(t)) | \left(\frac{dH(x(t))}{dt} - \frac{dE_0(x(t))}{dt} \otimes 1_{\E}\right) \phi_{E_0,\lambda}(x(t)) \rrangle}{\lambda(x(t))-\mu(x(t))} \simeq 0 \label{adiabcond}
\end{eqnarray}
Then the wave function of the universe, $\psi(t)$ solution of $\ihbar \frac{d\psi}{dt} = H(x(t)) \psi(t)$ with $\psi(0) = \phi^\alpha_{E_0,\lambda}(x(0))$, rests within the $*$-eigenspace associated with $E(x)$, i.e.
\begin{equation}
\forall t; \psi(t) \in \ker\left(H(x(t))-E_0(x(t)) \otimes 1_{\E} - \lambda(x(t)) 1_{\S \otimes \E}\right)
\end{equation}
\end{theo}

\begin{proof}
\begin{eqnarray}
& & \ihbar \frac{d \psi}{dt} = H \psi \nonumber \\
& & \iff \ihbar \frac{d}{dt} \left(\Te^{- \ihbar^{-1} \int_0^t E_0(t')dt'} \tilde \psi \right) = H \Te^{- \ihbar^{-1} \int_0^t E_0(t')dt'} \tilde \psi
\end{eqnarray}
with $\tilde \psi = \Te^{ \ihbar^{-1} \int_0^t E_0(t')dt'} \psi$, $\Te$ being the time ordered exponential (the Dyson series). We have then
\begin{eqnarray}
& & E_0 \Te^{- \ihbar^{-1} \int_0^t E_0(t')dt'} \tilde \psi + \ihbar \Te^{- \ihbar^{-1} \int_0^t E_0(t')dt'} \frac{d\tilde \psi}{dt} \nonumber \\
& & \qquad = H \Te^{- \ihbar^{-1} \int_0^t E_0(t')dt'} \tilde \psi
\end{eqnarray}
\begin{equation}
\Rightarrow \ihbar \frac{d\tilde \psi}{dt} = \Te^{\ihbar^{-1} \int_0^t E_0(t')dt'} (H-E_0) \Te^{- \ihbar^{-1} \int_0^t E_0(t')dt'} \tilde \psi
\end{equation}
\begin{equation}
\Rightarrow \ihbar \frac{d\tilde \psi}{dt} = \widetilde{(H-E_0)} \tilde \psi
\end{equation}
Let $\{\phi_{E_0,\lambda,a} \}_{\lambda \in \Sp(H-E_0),a \in \{1,...,\mathrm{deg}(\lambda)\}}$ be an orthonormal basis of the eigenspace of $H-E_0$ associated with $\lambda$. From the adiabatic condition we have $\forall \lambda \not= \mu, \forall a,b$
\begin{equation}
\llangle \phi_{E_0,\mu,a}|\partial_t|\phi_{E_0,\lambda,b} \rrangle = \frac{\llangle \phi_{E_0,\mu,a}|(\dot H - \dot E_0) \phi_{E_0,\lambda,b} \rrangle}{\lambda-\mu} \simeq 0
\end{equation}
$\Sp(\widetilde{H-E_0}) = \Sp(H-E_0)$ with orthonormal eigenbasis $\{\tilde \phi_{E_0,\lambda,a} = \Te^{\ihbar^{-1} \int_0^t E_0(t')dt'} \phi_{E_0,\lambda,a} \}$. We have then $\forall \lambda\not=\mu,\forall a,b$:
\begin{eqnarray}
\llangle \tilde \phi_{E_0,\mu,a}|\partial_t|\tilde \phi_{E_0,\lambda,b} \rrangle & = & \frac{\llangle \tilde \phi_{E_0,\mu,a}|(\dot{\tilde H} - \dot {\tilde E}_0) \tilde \phi_{E_0,\lambda,b} \rrangle}{\lambda-\mu} \nonumber\\
& = & \frac{\llangle \phi_{E_0,\mu,a}|(\dot H - \dot E_0) \phi_{E_0,\lambda,b} \rrangle}{\lambda-\mu} \nonumber \\
& & \quad + \ihbar^{-1} \frac{\llangle \phi_{E_0,\mu,a}|\overbrace{[E_0,H]}^{=0}|\phi_{E_0,\lambda,a}\rrangle}{\lambda-\mu} \nonumber \\
& \simeq & 0
\end{eqnarray}
We can then use the usual adiabatic approximation
\begin{equation}
\tilde \psi(t) = \sum_{a=1}^{\mathrm{deg}(\lambda)} c_a(t) \tilde \phi_{E_0,\lambda,a}(t)
\end{equation}
and then
\begin{equation}
\psi(t) = \sum_{a=1}^{\mathrm{deg}(\lambda)} c_a(t) \phi_{E_0,\lambda,a}(t) \in \ker(H-E_0-\lambda)
\end{equation}
\end{proof}
The use of the adiabatic approximation is here heuristic. The study of rigorous adiabatic theorems is not the subject of this paper; nevertheless this theorem shows that if the operator $H(x(t)) - E_0(x(t)) \otimes 1_{\E}$ satisfies a usual adiabatic theorem then the wave function of the universe stays in the $*$-eigenspace.

\begin{prop}
The adiabatic condition (\ref{adiabcond}) is equivalent to $R^\alpha_\mu(x(t)) \dot x^\mu(t) \simeq 0$ ($\forall t$) where $R^\alpha$ is the remainder of $C^*$-geometric phase generator.
\end{prop}

\begin{proof}
The set $\{\phi_{E_0,\lambda,a}\}_{\lambda \in \Sp(H-E_0),a =1,...,deg(\lambda)}$ constitutes an orthonormal basis of $\S \otimes \E$, we have then
\begin{equation}
 1 - P_{E_0,\lambda} = \sum_{\mu \not= \lambda} \sum_{a=1}^{deg(\mu)} = |\phi_{E_0,\mu,a}\rrangle \llangle \phi_{E_0,\mu,a}| 
\end{equation}
\begin{eqnarray}
R_\mu \dot x^\mu \|\phi_{E_0,\lambda}\|^2 & = & \tr_{\E}((1-P_{E_0,\lambda})|\partial_t \phi_{E_0,\lambda} \rrangle \llangle \phi_{E_0,\lambda}|) \nonumber\\
& = & \sum_{\mu \not=\lambda} \sum_a \llangle \phi_{E_0,\mu,a}|\partial_t \phi_{E_0,\lambda} \rrangle \nonumber \\
& & \qquad \times  \tr_{\E}(|\phi_{E_0,\mu,a} \rrangle \llangle \phi_{E_0,\lambda}|) \nonumber\\
& = & \sum_{\mu \not=\lambda} \sum_a \frac{\llangle \phi_{E_0,\mu,a}|(\dot H-\dot E_0 \otimes 1_{\E})|\phi_{E_0,\lambda} \rrangle}{\lambda-\mu} \nonumber \\
& & \qquad \times \tr_{\E}(|\phi_{E_0,\mu,a} \rrangle \llangle \phi_{E_0,\lambda}|) \nonumber \\
& \simeq & 0
\end{eqnarray}
\end{proof}

\subsection{Some properties of the adiabatic transport formula}
\begin{prop}
We have
\begin{equation}
\forall t \qquad \tr_{\S}\left(g_{EA}(t)\rho_E^\alpha(x(t)) g_{EA}(t)^\dagger \right) = 1
\end{equation}
\end{prop}

\begin{proof}
\begin{eqnarray}
\frac{d}{dt} \tr_{\S}(g_{EA}\rho_Eg_{EA}^\dagger) & = & \tr_{\S} (\dot g_{EA} \rho_E g_{EA}^\dagger) + \tr_{\S}(g_{EA} \dot \rho_E g_{EA}^\dagger) \nonumber \\
& & \qquad + \tr_{\S}(g_{EA}\rho_E \dot g_{EA}^\dagger)
\end{eqnarray}
but
\begin{equation}
\dot g_{EA} = - g_{EA}(A_\mu+\eta_\mu)\dot x^\mu  - \ihbar^{-1} E g_{EA}
\end{equation}
and
\begin{equation}
\dot \rho_E \simeq A_\mu \dot x^\mu \rho_E + \rho_E A^\dagger_\mu \dot x^\mu 
\end{equation}
then
\begin{eqnarray}
\frac{d}{dt} \tr_{\S}(g_{EA}\rho_Eg_{EA}^\dagger) %& = & - \tr_{\S}(g_{EA}(A_\mu+\eta_\mu)\dot x^\mu \rho_E g_{EA}^\dagger) \nonumber \\
%& & - \ihbar^{-1} \tr_{\S}(E g_{EA} \rho_E g_{EA}^\dagger) \nonumber \\
%& & + \tr_{\S}(g_{EA}(A_\mu \dot x^\mu \rho_E + \rho_E A^\dagger_\mu \dot x^\mu) g_{EA}^\dagger) \nonumber \\
%& & - \tr_{\S}(g_{EA}\rho_E (A^\dagger_\mu+\eta^\dagger_\mu)\dot x^\mu g_{EA}^\dagger) \nonumber \\
%& & + \ihbar^{-1} \tr_{\S}(g_{EA}\rho_E g_{EA}^\dagger E^\dagger) \nonumber \\
& = & - \tr_{\S}(g_{EA}(\eta_\mu \dot x^\mu \rho_E + \rho_E \eta^\dagger_\mu \dot x^\mu) g_{EA}^\dagger) \nonumber \\
& & - \ihbar^{-1} \tr_{\S}(E g_{EA} \rho_E g_{EA}^\dagger \nonumber \\
& & \qquad \qquad \qquad \quad - g_{EA} \rho_E g_{EA}^\dagger E^\dagger)
\end{eqnarray}
Moreover we have
\begin{equation}
\eta_\mu \dot x^\mu \rho_E + \rho_E \eta^\dagger_\mu \dot x^\mu = 0
\end{equation}
and
\begin{eqnarray}
& & g_{EA} \phi_E \in \ker(H-E) \nonumber \\
& & \Rightarrow \mathcal L(g_{EA}\rho_E g_{EA}^\dagger) = E g_{EA} \rho_E g_{EA}^\dagger - g_{EA} \rho_E g_{EA}^\dagger E^\dagger
\end{eqnarray}
and since
\begin{equation}
\tr_{\S}(\mathcal L(g_{EA}\rho_E g_{EA}^\dagger)) = \tr_{\S} \tr_{\E} ([H,g_{EA}|\phi_E \rrangle \llangle \phi_E|g_{EA}^\dagger ]) = 0
\end{equation}
we have
\begin{equation}
\tr_{\S}(E g_{EA} \rho_E g_{EA}^\dagger - g_{EA} \rho_E g_{EA}^\dagger E^\dagger) = 0
\end{equation}
Finally we have
\begin{equation}
\frac{d}{dt} \tr_{\S}(g_{EA}\rho_Eg_{EA}^\dagger) = 0
\end{equation}
\end{proof}
This result is the parallel transport condition for the $C^*$-geometric phases:
\begin{prop}
The parallel transport condition is
\begin{equation}
\frac{d}{dt}\left( \Pe^{-\int_{x(0)}^{x(t)} (\mathcal A(x)+\eta(x))} \rho_E(x(t)) \left(\Pe^{-\int_{x(0)}^{x(t)} (\mathcal A(x)+\eta(x))}\right)^\dagger \right)= 0
\end{equation}
\end{prop}

\begin{proof}
\begin{eqnarray}
\frac{d}{dt}(g_{\mathcal A} \rho_E g_{\mathcal A}^\dagger) & = & \dot g_{\mathcal A} \rho_E + g_{\mathcal A} \dot \rho_E g_{\mathcal A}^\dagger + g_{\mathcal A} \rho_E \dot g_{\mathcal A}^\dagger \nonumber \\
& = & - g_{\mathcal A} ({\mathcal A}_\mu+\eta_\mu)\dot x^\mu \rho_E g_{\mathcal A}^\dagger \nonumber \\
& & + g_{\mathcal A}({\mathcal A}_\mu \dot x^\mu \rho_E + \rho_E {\mathcal A}^\dagger_\mu \dot x^\mu ) g_{\mathcal A}^\dagger \nonumber \\
& & - g_{\mathcal A} \rho_E ({\mathcal A}^\dagger_\mu + \eta^\dagger_\mu)\dot x^\mu g_{\mathcal A}^\dagger \nonumber \\
& = & - g_{\mathcal A}(\eta_\mu \dot x^\mu \rho_E + \rho_E \eta^\dagger_\mu \dot x^\mu ) g_{\mathcal A} \nonumber \\
& = & 0
\end{eqnarray}
\end{proof}

\end{document}